\documentclass[11pt]{article}

\usepackage{tabularx}
\usepackage{diagbox}

\usepackage{amsmath,amssymb,graphicx,color,amsthm}
\usepackage[linesnumbered,ruled]{algorithm2e}
\usepackage{tikz}
\usepackage{dsfont}
\usepackage{mathtools}
\usetikzlibrary{shapes.geometric, arrows}
\usepackage{algorithmic}
\usepackage[english]{babel}
\usepackage[utf8]{inputenc}
\usepackage{mathrsfs}  
\usepackage{afterpage}
\usepackage{mathtools}
\usepackage{amsfonts}
\usepackage[colorinlistoftodos]{todonotes}
\usepackage{enumerate}
\usepackage[round]{natbib}
\usepackage{hyperref}
\newcommand{\excise}[1]{}

\newcommand\numberthis{\addtocounter{equation}{1}\tag{\theequation}}

\newcommand\RR{\mathbb{R}}

\newcommand{\cov}{\mathrm{Cov}}

\newcommand{\PD}{\mathrm{PD}}

\newtheorem{theorem}{Theorem}
\newtheorem{definition}{Definition}
\newtheorem{lemma}{Lemma}

\newtheorem*{example*}{Example}

\newtheorem{corollary}{Corollary}
\newtheorem{remark}{Remark}

\DeclarePairedDelimiterX{\infdivx}[2]{(}{)}{%
	#1\;\delimsize\|\;#2%
}

\usepackage{tikz}

\newcommand{\RNum}[1]{\uppercase\expandafter{\romannumeral #1\relax}}

\xdefinecolor{dukeblue}{rgb}{0.004,0.129,0.412}

\usepackage{times}
\usepackage{bm}
\usepackage{natbib}
\makeatletter
\newcommand*{\addFileDependency}[1]{
	\typeout{(#1)}
	\@addtofilelist{#1}
	\IfFileExists{#1}{}{\typeout{No file #1.}}
}
\makeatother

\usepackage{hyperref}

\usepackage{subcaption}
\usepackage{mathtools,amsfonts,amssymb}
\usepackage{enumerate}
\usepackage{comment}
\usepackage{mathrsfs}  
\usepackage{cleveref}

\DeclareMathOperator\GP{GP}

\usepackage{changepage}
\newif\ifhighlightchanges
\highlightchangestrue

\newcommand{\changedreviewerone}[1]{\ifhighlightchanges\textcolor{black}{#1}\else#1\fi}

\newcommand{\changedreviewertwo}[1]{\ifhighlightchanges\textcolor{black}{#1}\else#1\fi}

\newcommand{\changedreviewerthree}[1]{\ifhighlightchanges\textcolor{black}{#1}\else#1\fi}

\newcommand{\changedAE}[1]{\ifhighlightchanges\textcolor{black}{#1}\else#1\fi}

\newenvironment{review2_env}{\color{black}}{\color{black}}
\newenvironment{review3_env}{\color{black}}{\color{black}}

\begin{document}

	\title{\mbox{}\\[-11ex]Gaussian Processes for Time Series with Lead-Lag Effects with application to biology data}
	\author{Wancen Mu$^{1,*}$, Jiawen Chen$^{1,*}$, Eric S. Davis$^2$, Kathleen Reed$^3$,\\ Douglas Phanstiel$^{2,3,4}$, Michael I. Love$^{1,5}$, and Didong Li$^1$\\ 
		{\em Department of Biostatistics$^1$}\\
		{\em Curriculum in Bioinformatics and Computational Biology$^2$}\\
		{\em Curriculum in Genetics and Molecular Biology$^3$}\\
		{\em Department of Cell Biology and Physiology$^4$}\\
		{\em Department of Genetics$^5$, University of North Carolina at Chapel Hill} \\
		{\em Equal contribution$^*$}}
	\date{\vspace{-5ex}}	
	\maketitle

	\begin{abstract}
		Investigating the relationship, particularly the lead-lag effect, between time series is a common question across various disciplines, especially when uncovering biological process. However, analyzing time series presents several challenges. Firstly, due to technical reasons, the time points at which observations are made are not at uniform intervals. Secondly, some lead-lag effects are transient, necessitating time-lag estimation based on a limited number of time points. Thirdly, external factors also impact these time series, requiring a similarity metric to assess the lead-lag relationship. To counter these issues, we introduce a model grounded in the Gaussian process, affording the flexibility to estimate lead-lag effects for irregular time series. In addition, our method outputs dissimilarity scores, thereby broadening its applications to include tasks such as ranking or clustering multiple pair-wise time series when considering their strength of lead-lag effects with external factors. Crucially, we offer a series of theoretical proofs to substantiate the validity of our proposed kernels and the identifiability of kernel parameters. Our model demonstrates advances in various simulations and real-world applications, particularly in the study of dynamic chromatin interactions, compared to other leading methods. 
	\end{abstract}
	\section{Introduction}
	\label{sec:intro}
	The lead-lag effect, a widespread phenomenon where changes in one time series (the leading series) influence another time series (the influenced follower series) after a certain delay. This phenomenon is prevalent in a variety of fields, such as climate science \citep{environments8030018}, healthcare \citep{runge2019detecting, zhu2021evolution,feng2021two}, economics \citep{wang2017lead, skoura2019detection}\changedreviewerthree{, and finance~\citep{hoffmann2013estimation,bacry2013some,da2017correlation,ito2020direct}}. For example, \citet{harzallah1997observed} explores the lead-lag relationship between the Indian summer monsoon and various climate variables, such as sea surface temperature, snow cover, and geopotential height. However, this dynamic, rooted in causal regulatory relationships, is notably observed in biological processes, which serves as the primary inspiration for the proposed method. The lead-lag effect manifests itself in various processes, including development \citep{gerstein2010integrative, ding2020analysis, strober2019dynamic}, disease-associated genetic variation \citep{krijger2016regulation}, and responses to various interventions and treatments \citep{faryabi2008optimal, lu2021causal}. 
	
	To provide a clearer context for this concept, we consider temporal-omics datasets, particularly those used in the identification of potential enhancer-promoter pairs in various biological processes \citep{whalen2016enhancer, fulco2019activity, schoenfelder2019long,moore2020expanded}. 
	Enhancers are short cis-regulatory DNA sequences that are distal to the genes they regulate, as opposed to the promoter sequence that is adjacent to a particular gene. 
	Since enhancers and promoters coordinate to generate time- and context-specific gene expression programs, that is, the amount of RNA produced at what time, we hypothesize that the dynamic patterns of enhancer regulation are typically followed by corresponding gene program patterns, albeit with a certain time lag. 
	Approximately a million candidate human enhancers have been identified \citep{moore2020expanded}, 
	though the intricate dynamics of how and which enhancers interact with promoters in specific processes is still an area of active and ongoing research.

	Addressing this challenge, our study focuses on the multi-omics time series data presented by \citet{reed2022temporal}. This dataset includes measurements of enhancer activity, and promoter transcription, and 3D chromatin structure at seven time points, from 0 to 6 hours (a later time point is excluded from this analysis). These data reveal complex dynamic patterns, as depicted in \Cref{fig:introduction}. Our goal is to identify functional human enhancer-promoter pairs from over 20,000 candidate pairs, estimate the time lag, and quantify the extent of the lead-lag effect between enhancer activity and gene expression over time. \changedAE{The biological motivation for this goal is to uncover a relationship in genomics time series for a subset of enhancer-gene pairs, in particular those in which the enhancer state changes after cell activation and causes a subsequent change in expression of a target gene. Measuring the lag between enhancer activity and gene expression over time for these candidate enhancer-promoter pairs helps in understanding the sequence of regulatory events set off by activation, and how the time scale of gene regulation may vary across the genome.}
	
	\begin{figure}[h!]
		\centering
		\includegraphics[width=\textwidth]{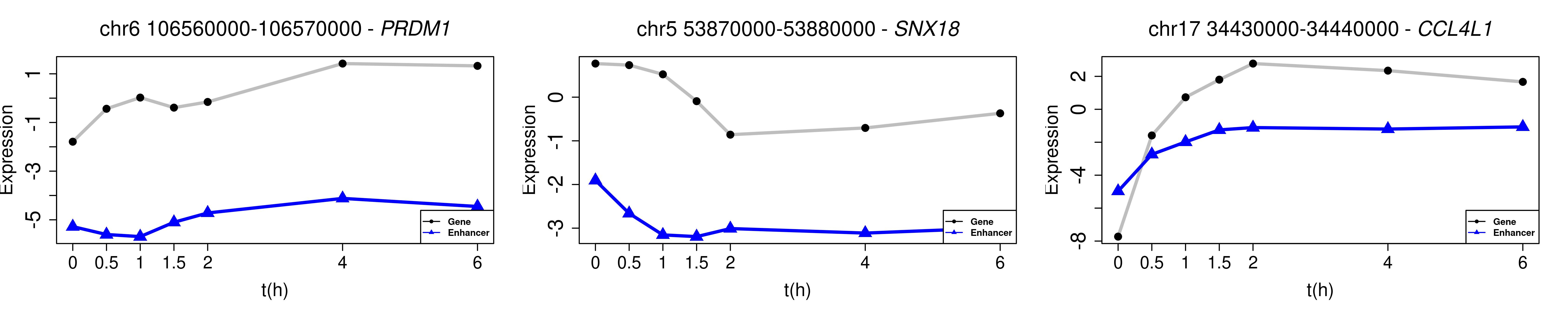}
		\caption{Examples of enhancer-promoter pairs with different dynamic patterns, The title for each subfigure represents the enhancer's genomic location  and the gene name.}
		\label{fig:introduction}
	\end{figure}

	
	In the context of our motivating application involving multi-omics time series data, and similar scenarios across various fields, there are key challenges that need careful consideration. Firstly, a common issue is the scarcity of time-point observations. This limitation is compounded by the irregular intervals between successive measurements, often due to the high cost of experimental design or sensor failures. Additionally, in cases where there is an abundance of time-point data, the focus may shift to pinpointing local time delays within smaller segments of the time series. These factors underscore the importance of accurately quantifying the uncertainty associated with the extent of the lead-lag relationship, particularly when dealing with limited temporal datasets.
	
	
	Existing methods for aligning time-lagged time series to explore the magnitude of lead-lag effects are numerous and varied. This exploration can extend to tasks such as ranking leaders or followers among multiple time series \citep{wu2010detecting} and performing network clustering of time series based on the pairwise lead-lag relationships \citep{bennett2022lead}. For example, \cite{wu2010detecting} ranks the leaders of the streams to the sea surface temperature. Widely used models such as time-lagged cross-correlation (TLCC,~\cite{shen2015analysis}) by comparing the cross-correlation coefficients, dynamic time warping (DTW, \cite{berndt1994using}), and its differentiable approximation, soft-DTW \citep{cuturi2017soft}, and soft-DTW divergence \citep{blondel2021differentiable} have proven useful by comparing distance loss among multiple pairwise time series after alignment~\citep{environments8030018}. TLCC, for instance, operates by sliding one signal over another and identifying the shift position with the highest correlation. However, these methods often require the same time intervals for each time series, which makes them less effective for irregular time series, potentially leading to distortions of the true underlying relationships. 
	
	Beyond quantifying the lead-lag effect between pairwise time series, there is often a desire to estimate time lag considering the irregularities in time series data, one approach is to first impute values at regular time points, then apply methods like TLCC to calculate the time lag. Common techniques for this separate modeling of each time series include using splines \citep{rehfeld2011comparison} or separate Gaussian processes (SGP, \cite{li2021multi}).
	However, such methods do not leverage shared information across time series, potentially leading to sub-optimal results. Hierarchical GP models such as Magma~\citep{leroy2022magma,leroy2023cluster} integrate shared information across multiple time series but are primarily designed for prediction and clustering individual time series. This is in contrast to our primary objective, which is to cluster pairwise time series by a measure of similarity. More direct approaches for irregular time series include the use of the negative and positive lead-lag estimator (NAPLES, \cite{ito2020naples}) and Lead-Lag method~\citep{hoffmann2013estimation}, which utilize the Hayashi–Yoshida covariance estimator to estimate a fixed time-lag. Alternatively, multi-task GPs through latent variable models introduce additional variables to account for temporal discrepancies such as GPLVM~\citep{durichen2014multitask}, which may limit flexibility due to the presumed independence between tasks and temporal sequences (also known as separable kernels). Moreover, more complex misalignment techniques \citep{kazlauskaite2019gaussian, mikheeva2022aligned} might not perform well with limited time observations for biology data where simpler linear mismatch assumptions suffice, and they can face operational challenges like Cholesky decomposition failures. \Cref{tab:methods} compares the aforementioned methods, highlighting their suitability for limited or irregular time points, their ability to quantify the magnitude of lead-lag effects, and estimate time lags:
	
	\begin{table}[!htbp]
		\centering
		\caption{An overview of related existing methods}
		\label{tab:tab1}
			\small 
			\begin{tabular}{|c|c|c|c|c|}
				\hline
				& \multicolumn{2}{c|}{Quantify the magnitude of lead-lag effect} & \multicolumn{2}{c|}{Estimate time lag}  \\ 
				\hline
				Model & \begin{tabular}[c]{@{}c@{}}Allows for limited\\ time points\end{tabular} & \begin{tabular}[c]{@{}c@{}}Allows for irregular\\ time points\end{tabular} & \begin{tabular}[c]{@{}c@{}}Allows for limited\\ time points\end{tabular} & \begin{tabular}[c]{@{}c@{}}Allows for irregular\\ time points\end{tabular} \\
				\hline
				TLCC  & X &         & X  &   \\
				\hline
				DTW      & X   &   & X  &    \\
				
				\hline
				Lead-Lag &   &         & X  & X  \\
				
				\hline
				SGP & X$^*$  & X$^*$       & X$^*$  & X$^*$   \\
				
				\hline
				Splines & X$^*$  & X$^*$  & X$^*$ & X$^*$  \\
				
				\hline
				GPLVM &   & X$^*$  &  & X$^*$  \\
				
				\hline
				MAGMA &  & X$^*$  &  &  X$^*$ \\
				
				\hline
				{\bf{GPlag}} & X & X & X & X\\
				\hline
			\end{tabular}\\
			{\small \emph{Note:} X$^*$ means yes if followed by TLCC or DTW.} 
			\label{tab:methods}
		\end{table}

		In response to these challenges and limitations of existing methods, we propose the \textbf{G}aussian \textbf{P}rocess for lead-\textbf{lag} time series (GPlag) model.
		GPlag introduces a novel class of GP kernels for time series with lead-lag effects, featuring interpretable parameters that estimate both the time lag and the degree of lead-lag effects. These parameters are crucial for tasks such as identifying associated factors, ranking or clustering time series based on lead-lag relationships. Our model also provides theoretical justification for the identifiability and interpretability of these key parameters, and integrates a Bayesian framework for greater flexibility in handling time series with constraints like limited observations or irregular time gaps.

		\section{Methods}\label{sec:model}
		\subsection{Gaussian Process}
		Gaussian processes are widely used for modelling dependency structures, such as spatial statistics and time series modelling. In generic terms, a Gaussian process models a random function where any finite collection of realizations (i.e., $n$ observations) follow a multivariate normal (MVN) distribution. 
		\begin{definition}
			$f:\Omega\to\RR$ follows a Gaussian process with mean function $\mu$ and covariance function $K$, denoted by $f\sim \mathrm{GP}(\mu,K)$, if for any $x_1,\cdots,x_n\in\Omega$,
			$[f(x_1),\cdots,f(x_n)]\sim N(v,\Sigma)$, where $v_i = \mu(x_i)$ and $\Sigma_{ij} = K(x_i, x_j)$.
		\end{definition}
		
		Covariance functions, also known as kernels, are widely studied when $\Omega\subset\RR^p$, such as the radial basis function (RBF, also known as the squared exponential kernel), the exponential kernel, and the Mat\'ern kernel~\citep{stein1999interpolation}. 
		
		\subsection{Pairwise Time Series Gaussian Process Model}\label{sec:pair}
		For the sake of clarity, we initiate pairwise time series models, and the extension to multiple time series is delayed to \Cref{sec:multi}. Given a pair of time series $Y(t)\coloneqq [Y_1(t),Y_2(t)]^\top$, where $t\in\RR$, we model $y$ as a multi-output Gaussian process: $Y=F(t)+\varepsilon$. Here, $F=[F_1,F_2]^\top\sim \GP(\mu,\widetilde{K})$ and $\varepsilon\sim N(0,\tau^2\mathrm{I}_2)$ represents the measurement error. The mean function is defined as $\mu$, and $\widetilde{K}:\RR\times\RR\to \PD(2)$ is a cross-covariance kernel, with $\PD(2)$ representing the space of all 2 by 2 positive definite matrices. $\mu$ is often assumed to be zero, which can be achieved by a vertical shift~\citep{gelfand2010handbook}, so we assume $\mu=0$ in the remaining sections.
		
		The cross-covariance kernel $\widetilde{K}$ is a crucial component of the model, expected to incorporate both the relationship between $Y_1$ and $Y_2$, and the time lag. Constructing valid cross-covariance kernels is a recognized challenge~\citep{gneiting2010matern,apanasovich2010cross,genton2015cross}, even more so when kernels must have certain identifiable and interpretable parameters in biological applications.
		
		We address this by transforming the problem to its ``dual" state, converting the multi-output GP into a single-output GP via $y=f(t,l)+\epsilon,~l=1,2$, where $f\sim N(0,K)$ is a single-output GP, $\epsilon\sim N(0,\tau^2)$, $F_1(t) = f(t,1)$ and $F_2(t)=f(t,2)$. Here, $K$ is a covariance kernel of a single-output GP, but the domain expands to $\RR\times\{1,2\}$ instead of $\RR$. Essentially, we transform a multi-output GP over $\RR$ into a single-output GP over $\RR\times \mathscr{C}$, where $\mathscr{C}=\{1,2\}$ is an index set of time series. This simplifies the problem by increasing the complexity of the domain."
		
		In this newly formed domain $\RR\times\mathscr{C}$, a family of inseparable, semi-stationary kernel has been \changedreviewertwo{developed} by \citep{li2021multi}, i.e., $K((t,l),(t',l'))=K((t+h,l),(t'+h,l'))$ for any $h\in\RR$. Consequently, $K$ can be reformulated as a function on $\RR\times\mathscr{C}\times \mathscr{C}$ as $K'(t,l,l')\coloneqq K((t,l),(0,l'))$. However, for simplicity, we will denote both as $K$ in the subsequent sections. A rich family of semi-stationary kernels can be found in \citep{li2021multi}. This family of kernels, although rich, was not specifically designed for time series with a time lag. To account for the time lag, the following theorem ensures that any existing semi-stationary kernel can induce a kernel with a parameter that measures the time lag, denoted by $s$, which \changedAE{ensures that our proposed kernel, with the inclusion of this time lag parameter, remains a valid kernel within the GP framework.}
		
		\begin{theorem}\label{thm:kernels}
			Let $K'$ be a semi-stationary kernel on $\RR\times \mathscr{C}$, then the induced kernel defined as
			\begin{equation}\label{eqn:kernels}
				K((t,l),(t',l'))\coloneqq K'(t-t'-\mathrm{\bf 1}_{\{l\neq l'\}}s,l,l')
			\end{equation}
			is positive definite, and thus serves as a covariance kernel for a GP, where $s\in\RR$ measures the time lag between two time series.
		\end{theorem}
		
		The following corollary provides some concrete candidates of kernels for GPlag:
		\begin{corollary}\label{cly:kernel_family}
			If $\phi: \RR_{\geq0}\to\RR$ is a completely monotone function and $\psi:\RR_{\geq0}\to\RR_{\geq0}$ is a non-negative function with a completely monotone derivative, then 
			$$K((t,l),(t',l'))=\frac{\sigma^2}{\psi(\mathrm{\bf 1}_{\{l\neq l'\}})^{1/2}}\phi\left(-\frac{|t-t'-\mathrm{\bf 1}_{\{l\neq l'\}}s|^2}{\psi(\mathrm{\bf 1}_{\{l\neq l'\}})}\right)$$
			is a valid kernel, where $\sigma^2>0$ measures the spatial variance, $s\in\RR$ measures the time lag.
		\end{corollary}
		The proof is given in \Cref{apdx:mtsgp}. Since the key idea in the above theorem and corollary is motivated by \cite{gneiting2002nonseparable}, we call the kernel family constructed from \Cref{cly:kernel_family} the Gneiting family. In particular,  Gneiting enumerates four candidates for $\phi$ and three for $\psi$ in \cite{gneiting2002nonseparable}. Exploring each amalgamation of $\phi$ and $\psi$ families allows for further tailoring by adjusting parameters intrinsic to the family, culminating in the desired kernel. Furthermore, more kernels can be constructed by combining kernels from Gneiting family and other simple kernels. Analogues of the RBF, exponential kernel, and Mat\'ern kernel in Euclidean space are provided in the following definitions. 
		
		\begin{definition}\label{def:LRBF}
			The following covariance kernel is called the time lag radial basis function (LRBF):
			\begin{equation}\label{eqn:LRBF}
				K((t,l),(t',l'))=\frac{\sigma^2}{
					(\mathrm{\bf 1}_{\{l\neq l'\}}a^2+1)^{1/2}}e^{-\frac{b\left(t-t'-\mathrm{\bf 1}_{\{l\neq l'\}}s\right)^2}{\mathrm{\bf 1}_{\{l\neq l'\}}a^2+1}},
			\end{equation}
		\end{definition}
		
		\begin{definition}\label{def:LExp}
			The following covariance kernel is called the time lag exponential kernel (LExp):
			\begin{equation}\label{eqn:LExp}
				K((t,l),(t',l'))=\frac{\sigma^2}{
					\mathrm{\bf 1}_{\{l\neq l'\}}a^2+1}e^{-b\left|t-t'-\mathrm{\bf 1}_{\{l\neq l'\}}s\right|},
			\end{equation}
		\end{definition}
		
		\begin{definition}\label{def:LMat}
			The following covariance kernel is called the time lag Mat\'ern kernel (LMat) with smoothness $\nu$:
			\begin{equation}\label{eqn:LMat}
				K((t,l),(t',l'))=\frac{\sigma^22\left\{\frac{b}{2}|t-t'-\mathrm{\bf 1}_{\{l\neq l'\}}s|\right\}^\nu }{(\mathrm{\bf 1}_{\{l\neq l'\}}a^2+1)^{\nu+1/2}\Gamma(\nu)}K_\nu\left(b|t-t'-\mathrm{\bf 1}_{\{l\neq l'\}}s|\right),
			\end{equation}
		\end{definition}
		
		
		As an illustration, LRBF can be derived by setting $\phi(t)=e^{-bt}$ and $\psi(t)=a^2t+1$ in \Cref{cly:kernel_family}; LExp is the product of a kernel in Gneiting family with $\phi(t)=e^{-bt^{1/2}}$, $\psi(t)\equiv 1$, and a simple kernel $\frac{1}{\mathrm{\bf 1}_{\{l\neq l'\}}a^2+1}$; LMat is the product of a kernel in Gneiting family with $\phi(t)=\frac{(bt^{1/2})^\nu}{2^{\nu-1}\Gamma(\nu)}K_\nu(bt^{1/2})$, $\psi\equiv 1$, and a simple kernel $\frac{1}{(\mathrm{\bf 1}_{\{l\neq l'\}}a^2+1)^{\nu+1/2}}$. \changedreviewertwo{Common choices for $\nu$ are $1/2,3/2,5/2$, where the Bessel function $K_\nu$ admits a closed form, similar to the standard Mat\'ern kernel, with the specific choice depending on the data. }
		Additional candidate kernels are presented in the Appendix A \Cref{apdx:kernels}, along with a detailed explanation on how they are constructed.
		
		\subsection{Interpretation of parameters}
		The interpretability of kernel parameters is one of the major advantages of GP models. In the kernels defined above, $\sigma^2$, the spatial variance, controls the point-wise variance; $b$, the lengthscale parameter, governs the temporal dependency; the smoothness parameter $\nu$ determines the smoothness of the sample path. In this paper, we assume $\nu$ to be given, which is a common assumption in GP literature~\citep{rasmussen2004gaussian}. More importantly, $a$, is called the dissimilarity parameter, where a smaller $a$ signifies a lesser degree of dissimilarity, which corresponds to a stronger lead-lag effect. For two distinct time series $l\neq l'$, when $a=\infty$, the covariance $K((t,l),(t',l'))=0$, indicating independence between time series, or no lead-lag effect at all. On the other hand, when $a=0$, $K((t,1),(t',2))=\sigma^2 e^{-b (t-t'-s)^2}$, implying that the two time series are the same up to a time lag, the strongest lead-lag effect. The interpretation of $a$ in GPlag parallels measures lead-lag effect in other classical methods like DTW and TLCC: $a = 0$ in GPlag aligns with a loss of 0 in DTW and a correlation of 1 in TLCC to signify the strongest lead-lag effect. Similarly, $a = \infty$ in GPlag corresponds to a loss of infinity in DTW and a correlation of 0 in TLCC for the weakest effect.  The time lag parameter $s$ measures the time lag between two time series. In the next subsection, we discuss the implementation of GPlag and the inference of these parameters, especially $a$ and $s$. \changedreviewerthree{Further visualization of sample paths of above kernels with different combinations of parameters are provided in \Cref{fig:sampp_exp}-\Cref{fig:sampp_matern}.} Theoretical support for the identifiability of $a$ and $s$ is provided in \Cref{sec:theory}.

		\subsection{Implementation}
		In this section, we outline the algorithm to estimate the parameters $\{a,b, s,\sigma^2,\tau^2\}$ proposed in previous model, given the observations $\{(t_i,c_i, y_i)\}_{i=1}^{n}, c_i\in \{1, 2\}, n = n_1+n_2$ . The most direct approach is to find the maximum likelihood estimator (MLE) given by:
		\begin{align}
			l(a,b,s,\sigma^2,\tau^2)&\coloneqq \log~p(y_1,\cdots,y_{n}|c_1,\cdots, c_n, t_1,\cdots,t_{n},a,b,s,\sigma^2,\tau^2)\nonumber \\
			&= - \frac{n}{2}\log2\pi -\frac{1}{2}\log|\Sigma+\tau^2\mathrm{I}_{n}| - \frac{Y^\top (\Sigma+\tau^2\mathrm{I}_{n})^{-1}Y}{2},\label{eqn:mll}
		\end{align}
		
		where $Y=[y_1,\cdots,y_{n}]^\top$ denotes the vector of observations, and $\Sigma$ is the $(n_1+n_2)\times (n_1+n_2)$ covariance matrix specified in \Cref{eqn:LRBF,eqn:LExp,eqn:LMat}.
		Optimization methods such as L-BFGS-B~\citep{byrd1995limited} and Adam ~\citep{kingma2014adam} are widely used. The L-BFGS-B algorithm allows for setting lower and upper bounds on the parameters, enhancing the stability of the numeric optimization. The parameters $a,~b,~\sigma^2,~\tau^2$ are constrained to be positive and $s$ can be set to a reasonable range depending on specific cases. For the initial values of $b,~\sigma^2,~\tau^2$, we suggest running a standard GP regression first and using the estimates as our initial values. For $s$, we recommend taking the estimate from TLCC, and we recommend $a=1$ as the initialization. The illustration of MLE process is provided in Appendix A \Cref{alg:GPlag}.
		
		The extension to Bayesian inference is straightforward: we can simply sample from the posterior $\pi(a,b,s,\sigma^2,\tau^2|\{t_i,c_i, y_i\}_{i=1}^{n})\propto \pi(a,b, s,\sigma^2,\tau^2)p(Y|\{y_i, c_i,t_i\}_{i=1}^{n}, a,b,s,\sigma^2,\tau^2)$, where $\pi(a,b, s,\sigma^2,\tau^2)$ is the joint prior distribution. We suggest using independent inverse Gamma priors for $a,~b,~\sigma^2,~\tau^2$ and a Gaussian prior for $s$. \changedreviewerthree{The algorithm is depicted in \Cref{alg:GPlag-bayesian} in \Cref{apdx:algorithm} and has been implemented in Python using GPyTorch~\citep{gardner2018gpytorch} and the Pyro~\citep{bingham2019pyro} package.} This implementation enables rapid, fully Bayesian inference and takes advantage of GPU acceleration.



		\subsection{Extension to multi-time series}\label{sec:multi}
		The kernels introduced in Section \ref{sec:pair} can be extended to accommodate any number of time series, denoted by $L\geq 2$. Following the same rationale, it suffices to define a kernel on $\RR\times \mathscr{C}\times \mathscr{C}$ where $\mathscr{C}=\{1,2,\cdots,L\}$. The following theorem generalizes LRBF, LExp and LMat for any $L$ by aligning \changedreviewertwo{$A=(a_{ll'})$} with the criteria for a distance metric, parameterized by $a_{ll'}\geq 0,b,s_l,\sigma^2,\tau^2$ under the constraints $a_{ll'}=0\Longleftrightarrow l=l',~a_{ll'}=a_{l'l},~a_{ll'}+a_{l'k}\geq a_{lk}$ for any $l,l',k=1,\cdots,L$ and $s_1=0$:
		\begin{theorem}\label{thm:multiGPlag}
			The following covariance functions are extensions of LRBF, LExp and LMat:
			\begin{align*}
				K((t,l),(t',l'))&=\frac{\sigma^2}{
					(a_{ll'}^2+1)^{1/2}}e^{-\frac{b\left(t-t'+s_l-s_{l'}\right)^2}{a_{ll'}^2+1}},\\
				K((t,l),(t',l'))&=\frac{\sigma^2}{
					(a_{ll'}^2+1)}e^{-b\left|t-t'+s_l-s_{l'}\right|},\\
				K((t,l),(t',l'))&=\frac{\sigma^22\left\{\frac{b}{2}|t-t'+s_l-s_{l'}|\right\}^\nu }{(a_{ll'}^2+1)^{\nu+1/2}\Gamma(\nu)}K_\nu\left(b|t-t'+s_l-s_{l'}|\right).
			\end{align*}
		\end{theorem}
		The positive definite, symmetric, and triangle inequality constraints on $a_{ll'}$ ensures the positive definiteness of $K$, while the constraint on $s_1$ is due to our treating the first time series as the baseline, meaning that the time lag is relative to the first time series. The inference is similar to the previous case, with the only difference being that the optimization method involved in MLE needs to consider these additional constraints. In other words, it is a constrained optimization problem. Nevertheless, it is not intractable, as all the constraints on $a_{ll'}$ are linear and can be accommodated by existing optimization methods, including L-BFGS-B and Adam, as illustrated in Appendix A \Cref{alg:GPlag-multi} with three time series.
		
		\section{Identifiability of GPlag Kernel Parameter}\label{sec:theory}
		
		To support our interpretation to $A\coloneqq(a_{ll'})\in\RR^{L\times L}$ and $S\coloneqq (s_l)\in\RR^{L}$ in Section \ref{sec:model}, we discuss the identifiability, i.e., different parameters lead to different models. Note that not all parameters of commonly used GP models are identifiable, which means that different parameter values can lead to the same model. For example, in the classical Mat\'ern covariance function, the spatial variance $\sigma^2$ and lengthscale parameter $b$ are not identifiable, meaning that it is not possible to distinguish between different values of these parameters based on the data alone~\citep{stein1999interpolation,zhang2004inconsistent,kaufman2013role,li2021inference,li2022bayesian}. This can be a limitation of GP models and it is important to consider when interpreting the results. In the case of GPlag, if the key parameters $A$ and $S$ are not identifiable, i.e., the GPlag determined by $(A,S)$ is equivalent to the one determined by $(\widetilde{A},\widetilde{S})$, then it is unconvincing to interpret $A$ and $S$ as times series dissimilarity and time lag. Without identifiability, $A$ and $S$ do not admit any clear interpretation and the results of the model would be hard to interpret. 
		
		\changedAE{In this section, we prove that the parameters $A$ and $S$ in GPlag model are identifiable, allowing them to be used to as measures of the lead lag effect between time series (measured by $A$) and the time lag (measured by $S$).} \changedAE{Here, we assume all parameters are positive, finite, real numbers, and the domain is fixed (also known as infill domain).} \changedAE{The fixed domain assumption is reasonable for multi-omics time series, as the dynamics of gene regulation in such experiments occur during a fixed window of time. Investigators will decide, based on available budget, the time points needed to capture the dynamics during, say activation, or differentiation. If allocated more budget, they will then fill in the gaps between existing observations, to have better resolution over the same time period.} We begin by introducing the necessary definitions.
		
		\begin{definition}
			Two Gaussian processes $K$ and $\widetilde{K}$ are said to be equivalent, denoted by $K\equiv \widetilde{K}$ , if the corresponding Gaussian random measures are equivalent to each other. That is, two Gaussian random measures are absolutely continuous with respect to each other. 
		\end{definition}
		
		If $K\equiv \widetilde{K}$, then $K$ can never be correctly distinguished from $\widetilde{K}$ regardless how dense the observed time points are~\citep{zhang2004inconsistent}. Focusing on a parametric family of covariance functions $K_\theta$, if there exists $\theta_1\neq \theta_2$ such that $K_{\theta_1}\equiv K_{\theta_2}$, then any estimator of $\theta$ based on $n$ observations $\{t_i,y_i\}_{i=1}^n$, denoted by $\widehat{\theta}_n$, cannot be weakly consistent \citep{dudley1989real}. 
		
		\begin{definition}
			$\theta$ is said to be identifiable if $K_\theta\equiv K_{\widetilde{\theta}}\Longleftrightarrow \theta=\widetilde{\theta}$. 
		\end{definition}
		
		\changedreviewerthree{As a result, it suffices to show the following necessary and sufficient conditions for two GPlags being equivalent, which suggests identifiable parameters for further interpretation:}
		
	
	\begin{theorem}\label{thm:iden}
		Let $K$ and $\widetilde{K}$ be two LMat kernels with parameters $(\sigma^2,b,A,S)$ and $(\widetilde{\sigma}^2,\widetilde{b},\widetilde{A},\widetilde{S})$ and $\nu$ is assumed to be known, then 
		\begin{equation*}
			K\equiv \widetilde{K} \Longleftrightarrow (\sigma^2b^{2\nu},A,S)=(\widetilde{\sigma}^2\widetilde{b}^{2\nu},\widetilde{A},\widetilde{S}).
		\end{equation*}
		That is, the identifiable parameters, also known as the microergodic parameters, are $\{\sigma^2b^{2\nu}, A, S\}$. In particular, let $\nu=\frac{1}{2}$, the microergodic parameters for LExp kernel are $\{\sigma^2b,A,S\}$.
	\end{theorem}
	The proof is given in Appendix \ref{apdx:iden}. Although $\sigma^2$ and $b$ in GPlag are not identifiable, the proposed interpretation of the model does not rely on these two parameters. Instead, the parameter $A$ is interpreted as a measure of dissimilarity between time series under a time lag represented by the parameter $S$.
	
	It is important to note that GPs are commonly used for prediction and regression tasks. For this purpose, even though some kernel parameters such as $\sigma^2$ and $b$ are not identifiable in these models, the prediction performance remains asymptotically optimal when the kernel parameters are mis-specified~\citep{kaufman2013role}. This is because regression is a distinct problem from parameter inference in GPs.
	
	It is also worth noting that the identifiability of kernel parameters in GPs depends on the dimension of the domain. For time series, the domain is a 1-dimensional interval in $\RR^1$. However, our theorem holds for 2-dimensional and 3-dimensional domains. In dimensions greater than $3$, it becomes a challenging problem and we treat it as a future work since our focus of this work is one-dimensional time series.
	
	\section{Simulation studies}\label{sec:sim}
	The main focus of our experiments was to illustrate that GPlag could accurately estimate both the time lag and the extent of lead-lag effects. This capability aided in the ranking or clustering of multiple time series that existed in a lead-lag relationship with a target time series. We demonstrated the theoretical properties of GPlag using synthetic data, as these properties were best illustrated under a fully controlled environment. All code and additional implementation details were available in the \Cref{apdx:simulation}.
	
	\textbf{Baseline models} 1) For estimating the time lag, we compared with Lead-Lag \citep{hoffmann2013estimation}. 2) By utilizing the GP framework, we were able to predict values for two time series at unobserved time points. Then we evaluated the prediction error against three baselines: natural cubic splines with five
	breakpoints, SGP and Magma. 3) To evaluate the effectiveness of the estimate of lead-lag effects, we employed four baseline methods: TLCC, DTW, soft-DTW, and soft-DTW divergence. TLCC and DTW could output aligned time series, allowing us to calculate correlation coefficients from the aligned data. On the other hand, soft-DTW and soft-DTW divergence measured the amount of warping needed to align two time series and provided a distance loss metric. We used both correlation coefficients and distance loss as measures of dissimilarity between the time series.
	
	\textbf{Parameter estimation and inference.}
	To validate the identifiability of our model, we simulated a pair of time series from GPlag with three kernels in \Cref{eqn:LRBF,eqn:LExp,eqn:LMat}, with 100 replicates. We set the mean to zero for both series and specified parameters as $b = 0.3,~a = 1,~s = 2,~\sigma^2=4$ for each kernel. For LMat, we set the smoothness parameter as $\nu = 3/2$. The number of measurements within each time series varied across $ \{20, 50, 100, 200\}$ and the time points were generated following $t_i = i + \epsilon$ and $t'_i = t_i + s$, where $i \in \{1,2, \dots, n\}$ and $\epsilon \sim \text{Unif}(-\frac{1}{4}, \frac{1}{4})$ represents noise.
	
	\begin{figure}[!htbp]
		\vskip -0.1in
		\begin{center}
			\centerline{\includegraphics[width=\columnwidth]{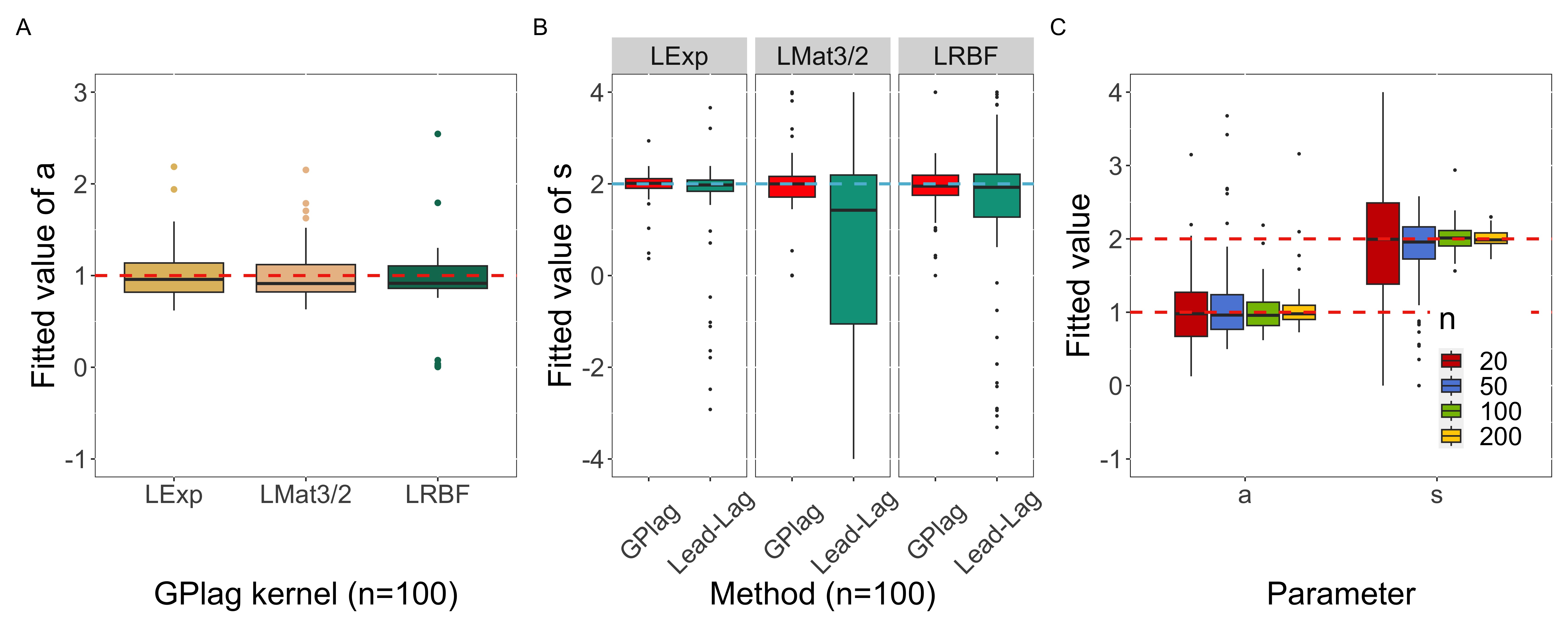}}
			\caption{Parameter estimation. (A) MLE of $a$ in GPlag with data ($n=100$) generated from three different kernels. (B) $s$ estimated by GPlag and Lead-Lag with data ($n=100$) generated from three different kernels. (C) MLEs of GPlag parameters $a$ and $s$ in the LExp kernel with different $n$. The true parameter values were represented by the dashed red horizontal lines. All boxplots were from 100 replicates.}
			\label{fig:identifibility}
		\end{center}
		\vskip -0.2in
	\end{figure}
	
	Our results showed that the MLE of $a$ and $s$ were very close to their truth among all three kernels when $n = 100$. This was consistent with the theoretical analysis in \Cref{thm:iden} (\Cref{fig:identifibility} A, B) that GPlag could accurately recover the true lead-lag effects. Compared to Lead-Lag, GPlag exhibited similar or smaller variance, specifically when the kernel was Mat\'ern with $\nu=3/2$. 
	Additionally, we performed simulations with a varying number of time points and found that the variance of the MLEs of $a$ and $s$ decreased as the sample size increased, which empirically confirmed the consistency of the MLEs (\Cref{fig:identifibility} C). \changedreviewertwo{Plots for other parameters of secondary interest, including $\sigma^2$, $b$, and $\sigma^2b^{2\nu}$, are provided in \Cref{fig:bsigma2}. We also performed simulations with different settings of $b=0.3,5,10$, and found that the consistency of the MLEs for $a$ and $s$ was maintained across these different values of $b$ (See \Cref{fig:variousb} for more details).}

	\textbf{Prediction.}\label{subsec:pred}
	We assessed the accuracy of GPlag's predictions by generating data from two different processes. In the first, we used the kernel LExp in \Cref{eqn:LExp} with parameters $b = 1,~a = 0.3,~s = 2$, and $\sigma^2 = 4$. In the second, we generated pairwise time series from linear functions with non-Gaussian noise (t-distributed): $y = 2t + 3 + 5\epsilon_1$ and $y = 2(t-20) + 3 + 5\epsilon_2$ where $t =0, 1, 2, \ldots, 100$. The noise term, $\epsilon_1, \epsilon_2$, were drawn from a t-distribution with 5 degrees of freedom. We then randomly selected $50\%$ of the data as training data for model fitting and used the remaining $50\%$ as testing data for prediction in both settings. 
	We applied the best linear unbiased prediction (BLUP, see \Cref{apdx:simulation} for more details) and used the mean squared error (MSE) on the testing set as the metric to evaluate the prediction accuracy. 
	
	Among all methods, GPlag achieved the highest log-likelihood and smallest MSE~(\Cref{fig:ranking} A-B, \Cref{fig:R2Q2}). The advantage of GPlag over SGP and splines was that it modeled two time series together which allowed for sharing of information between time series through $a$. When $a \rightarrow \infty$, it was equivalent to SGP. While a small estimated $a$ indicated that two time series were similar, and the covariance function then allowed for more information to be shared between them. Although Magma adopted a hierarchical GP to specifically model the mean function for irregular data, it did not consider time lag information across time series.

	\textbf{Clustering or ranking time series based on dissimilarity to target time series.}

	We evaluated the relationship between the magnitude of parameter $a$ and the degree of lead-lag effects (the level of synchrony with a designated target signal) to assess the effectiveness of GPlag. We simulated nine time series from the function $f_k(t) = \frac{\arctan(k(t + s))}{\arctan(k)}$ for every pair of $k$ and $s$, where $k\in \{0.01, 1, 10\}$ and $s\in \{0, 0.5, 1\}$. Among them, the time series simulated from the function of $k=0.01,~s=0$ was treated as the target time series, which was close to a linear function~(\Cref{fig:ranking} C, ts1). Each time series had $50$ times points ranging from $[-2,2]$. Here, $k$ determined the degree of distortion, where a smaller $k$ (e.g. 0.001) corresponded to a curve close to the identity transformation, and a larger $k$ represented the cases with more upward/downward distortion. By comparing different columns in \Cref{fig:ranking} C, we could see that the degree of dissimilarity between the target time series and other time series was controlled by $k$, whereas $s$ controlled the time lag along the $t$-axis.
	
	\begin{figure}[htb]
		\vskip -0.1in
		\begin{center}
			\centerline{\includegraphics[width=5.5in]{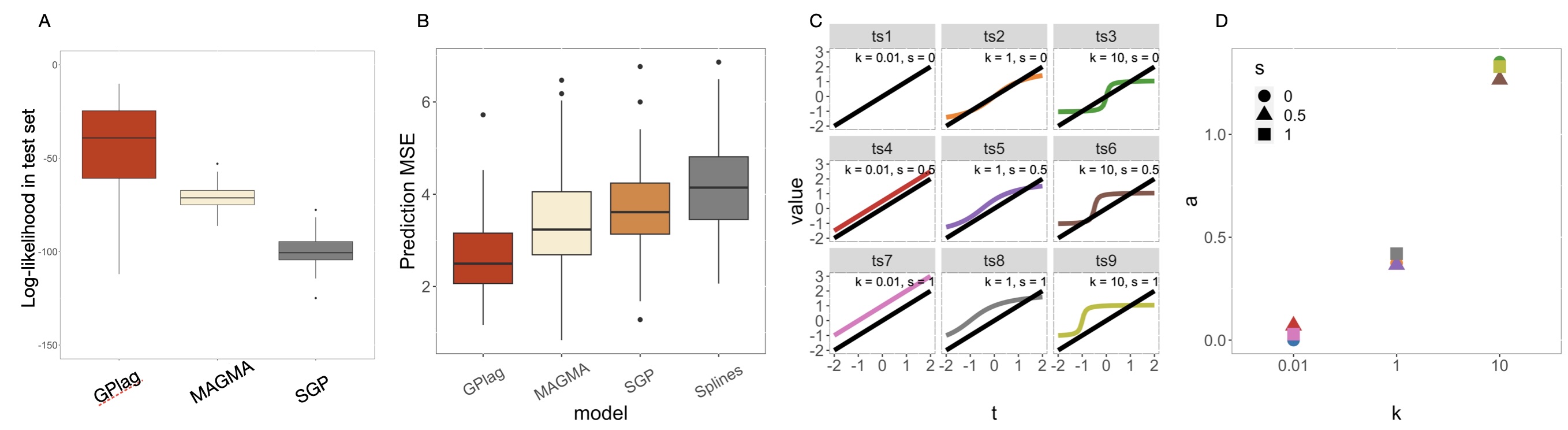}}
			\caption{(A-B) Evaluation of Prediction Performance for Time Series Derived from LExp Kernel, as Assessed by Log-Likelihood (A) and Mean Squared Error (MSE) (B). (C) Line plots of 9 features time series vs. target time series with various orders of distortion and time lag. The black line represents the target time series. (D) Estimation of $a$ in 9 pairwise comparisons. }
			\label{fig:ranking}
		\end{center}
		\vskip -0.2in
	\end{figure}

	Nine pairs of time series were fitted by GPlag with the LExp kernel. We observed that three pairs of time series in each column with the same $k$ were grouped together (\Cref{fig:ranking} D), and $a$ increased as $k$ increased.
	This implied that GPlag could accurately rank the degree of lead-lag effects among multiple pairs of time series based on the parameter $a$, through precise estimation of the time lag parameter $s$. A larger estimate of $a$ indicated a stronger dissimilarity between pairs of time series. To further quantify the clustering performance of lead-lag effects among the nine pairs of time series, we used the Adjusted Rand Index (ARI,~\cite{hubert1985comparing}) and Normalized Mutual Information (NMI,~\cite{strehl2002normalized}) as evaluation metrics. To define a cluster from the pairwise dissimilarity parameter (correlation coefficients for TLCC and distance loss for DTWs), we used K-means~\citep{6773024} with 3 clusters. As the focus of this work was on determining  pairwise dissimilarity rather than the clustering algorithm, we opted for a simple clustering method. GPlag achieved the highest ARI and NMI in clustering the lead-lag effects among 9 pairwise time series (\Cref{tab:tab1}, column 2-3). \changedreviewertwo{We also tested GPlag using the LRBF and LMat kernels, and found that ARI and NMI scores remained at 1, consistent with the LExp results. The inferred $a$ values are shown in \Cref{fig:cluster_rbf_mat}, further indicating that the results are not sensitive to the choice of kernel.}

	\begin{table}[h]
		\centering
		\caption{Evaluation of GPlag vs TLCC, DTW, soft-DTW, soft-DTW divergence on various datasets}
		\label{tab:tab2}
		\small
			\begin{tabular}{|c|cc|cc|}
				\hline
				& \multicolumn{2}{c|}{Synthetic data clustering} & \multicolumn{2}{c|}{Dynamic Chromatin Interactions}  \\ 
				\hline
				Model & \multicolumn{1}{c|}{ARI}   & NMI         & \multicolumn{1}{c|}{\begin{tabular}[c]{@{}c@{}}p-value: \\ association test\end{tabular}} & \begin{tabular}[c]{@{}c@{}}p-value: \\ enrichment test\end{tabular} \\ 
				\hline
				TLCC  & \multicolumn{1}{c|}{0.35} & 0.58        & \multicolumn{1}{c|}{0.051}  & 0.13  \\ 
				\hline
				DTW      & \multicolumn{1}{c|}{0.4}   & 0.55  & \multicolumn{1}{c|}{0.292}  & 1   \\ 
				\hline
				soft-DTW & \multicolumn{1}{c|}{0.27}  & 0.49        & \multicolumn{1}{c|}{0.957}  & 0.57   \\ 
				\hline
				soft-DTW divergence & \multicolumn{1}{c|}{0.07}  & 0.39        & \multicolumn{1}{c|}{0.907}  & 0.90   \\ 
				\hline
				GPlag & \multicolumn{1}{c|}{\textbf{1.00}}  & \textbf{1.00}  & \multicolumn{1}{c|}{\textbf{0.019}} & \textbf{0.05}  \\ 
				\hline
			\end{tabular}%
		\\
		{\small \emph{Note:} Since the above algorithms are deterministic, the standard deviation is not reported.}
	\end{table}

	
	\textbf{Extension to three time series.}
	To demonstrate the applicability of GPlag to more than two time series, we conducted simulation studies using three time series generated from the LExp kernel. In this scenario, all constraints were linear, allowing us to utilize optimizers that support linear constraints. Similar to the pairwise simulation setting, we set the mean to zero for all series, and the time points were generated according to the equations $t_i = i + \epsilon$ and $t'_i = t_i + s_i$, where $i \in \{1,2, \dots, 50\}$ and $\epsilon \sim \text{Unif}(-\frac{1}{4}, \frac{1}{4})$ represented noise. The specified parameter values were $b = 0.3, a_{12} = a_{13} = a_{23} = 1, s_2 = 2, s_3 = 4, \sigma^2 = 4$. The maximum likelihood estimate (MLE) of the parameters closely approximated the true values (represented by the red dashed line), providing further validation of the theoretical properties. Additionally, the proposed method achieved the lowest mean squared error (MSE) in this particular setup (see \Cref{fig:multitimeseries}).
	
	\begin{figure}[h!]
		\centering
		\includegraphics[width=0.5\textwidth]{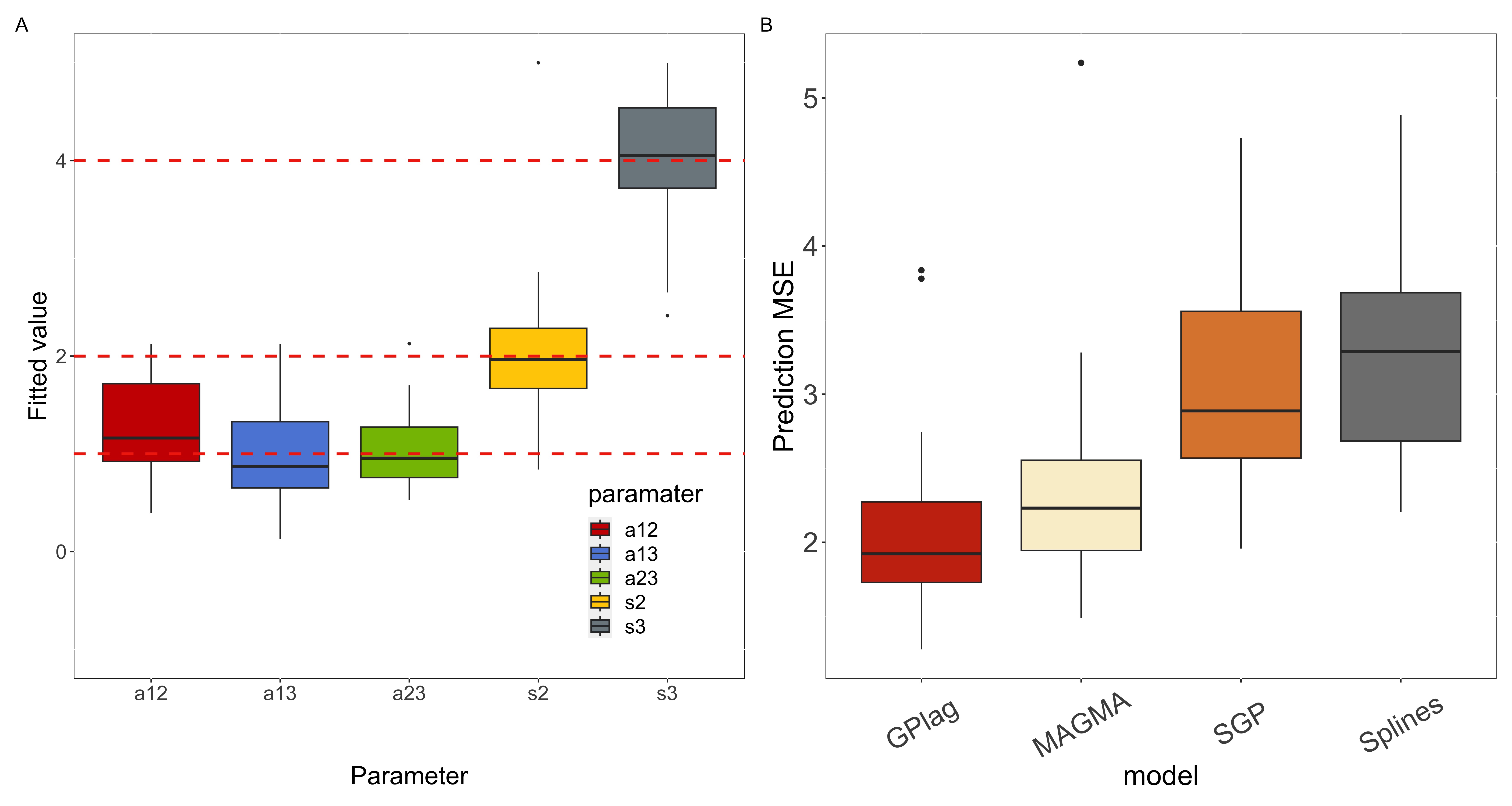}
		\caption{Fitting performance on three time series generated from LExp with b = 0.3, $a_{12}=a_{13}=a_{23}=1$, $s_2=2, s_3=4$. (A) MLE of $a_{12},a_{13},a_{23},s_1,s_2$. The true parameter values
			are represented by the dashed horizontal lines. (B) Prediction MSE on the held-out dataset, from 100 replicates.}
		\label{fig:multitimeseries}
	\end{figure}

	\section{\changedAE{Genomics application}: Dynamic chromatin interactions}\label{sec:appl}
	In this section, we applied GPlag to a time-resolved multi-omics dataset involving enhancers and genes in human cells activated by interferon gamma from a resting state \citep{reed2022temporal}. Originally, the dataset comprised over 20,000 pairs. Each pair consisted of one time series for each enhancer and likewise one time series for each gene, at seven irregularly-spaced time points (0, 0.5, 1, 1.5, 2, 4, and 6 hours after a perturbation). 
	The time series captured gene expression measurements for genes and histone tail acetylation measurements (H3K27ac), a measure of activity for the enhancers, both approximately variance stabilized.
	It is widely believed that the alterations in acetylation occur prior to the changes in expression of the regulated genes \citep{reed2022temporal}.
	\changedAE{We focused only on pairs that became functional after stimulation, by applying a filter on the total range of stabilized values across the timescale, yielding 4,776 dynamic pairs of enhancer-gene time series.
		Subsequently, the objective of our analysis, after estimating the model parameters of GPlag for these 4,776 pairs, was to identify the subset of all possible enhancer-gene pairs that represent \textit{functional} enhancer-gene pairs, i.e. the enhancer responds to cell activation and regulates the expression of the gene, after some time lag. Many types of time profiles can be captured by such a time series experiment, e.g. increasing across the range, decreasing, or increasing then decreasing, etc. Particular functional enhancer-gene pairs will have subtle differences in the dynamics of the coordinated change in activity and expression, compared to other functional pairs, which we hope to capture with our GPlag model.}
	
	
	
	
	To assess the performance of GPlag in ranking enhancer-gene connections, we split the pairs into two groups according to the parameter $a$. The top 20\% pairs were designated as the candidate enhancer-gene group, with the remaining pairs forming the negative group. To compare against baselines, we split all pairs into two groups using their respective metric of similarity and a 20\% cutoff. Subsequently, we assessed the performance from three perspectives. 
	
	Firstly, we compared the Hi-C (High-throughput Chromosome Conformation Capture, \cite{lieberman2009comprehensive}) counts between the candidate enhancer-gene group and the negative group. Since Hi-C counts measured the frequency of physical interactions between enhancers and genes in three-dimensional chromatin structure, we expected the candidate enhancer-gene group to have higher Hi-C counts, 
	while the Hi-C information was not used in our ranking of pairs.
	Our results showed candidate enhancer-gene pairs exhibited significantly higher Hi-C counts compared to the other pairs~(\Cref{tab:tab2}, column 4-5), with GPlag producing the smallest $p$-value (Wilcoxon rank-sum test, $p = 0.019$). 
	
	Secondly, we analyzed the enrichment of \textit{looped} enhancer-gene pairs between the two groups. A \textit{looped} enhancer-gene pair is one where the enhancer element is physically brought into close proximity with the regulated gene, enabling the enhancer's activity to influence the expression of the associated gene. Looping in this dataset was determined using Hi-C counts, but the determination of loops uses spatial pattern detection and statistical models to determine if the signal rises to the level of a ``loop''.
	We observed that candidate enhancer-gene pairs identified by GPlag had a significantly higher likelihood of being ``looped" enhancer-gene pairs compared to the matched pairs from the other group, selected based on genomic distance \citep{davis2022matchranges} ($\chi^2$ test, $p = 0.05$). However, no significant difference in enrichment was found between the two groups of pairs identified using baseline methods.
	
	The lack of statistical significance observed in DTW-based methods may be due to the limited information available for alignment because of the small number of observations. Additionally, with irregularly spaced time points, TLCC might introduce bias when estimating time lag and coefficients. Given the constraint of TLCC, which can only generate time lag estimates from the observed time points, its performance was compromised when applied to this dataset with irregular time gaps (\Cref{fig:realdata}). However, GPlag, by estimating a continuous time lag, was able to identify lags that restored strong similarity between the pairwise time series. 

	Thirdly, to evaluate if our candidate enhancer-promoter pairs were detected as having a functional relationship in an orthogonal experimental data type, we downloaded the genomic locations of response expression quantitative trait loci (response eQTL) found in human macrophages exposed to IFN$_{\gamma}$ \citep{alasoo2018shared} (that is, using one of the two stimuli used in the \citet{reed2022temporal} experiment, and in the same cell type). Because the number of eQTLs was small compared to the number of enhancers, we used a relatively stringent threshold to label a putative functional enhancer-promoter pair: 
	we split the candidate enhancer-promoter pairs into groups according to posterior mean correlation ($> 0.98$). There were 74 H3K27ac peaks overlapping with macrophage response eQTL, among which 12 pairs were in our defined candidate enhancer-promoter pair sets. Therefore, there was a significant enrichment of our candidate enhancer-promoter pairs to eQTLs ($\chi^2$ test $p = 0.18 \times 10^{-3}$) futher suggesting the relevance of the enhancer-promoter pairs highly ranked by GPlag.

	\begin{figure}[!h]
		\centering
		
			\centering
			\includegraphics[width=\textwidth]{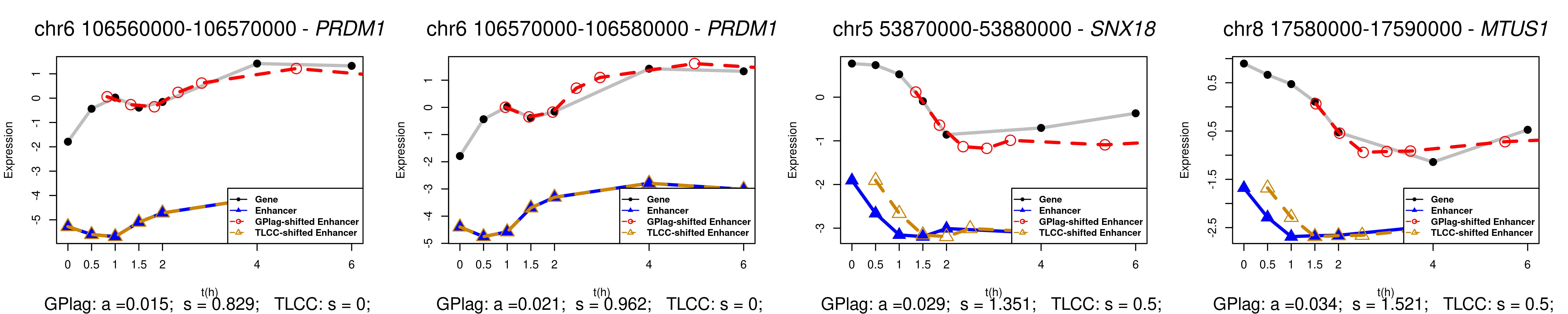}
			\label{fig:plot1}
		
		
		\caption{Visualization of GPlag modelling output on selected pairs. The title for each subfigure represents the enhancer's genomic location  and the gene name. 
			The solid circle line represents gene expression; 
			the solid triangle line represents enhancer activity; 
			the dashed circle line represents the enhancer shifted by GPlag; and 
			the dashed triangle line represents the enhancers shifted by TLCC. }
	\vspace{-13pt}
	\label{fig:realdata}
\end{figure}

\section{Discussion}\label{sec:disc}
\vskip -5pt
GPlag is a flexible and interpretable model specifically designed for irregular time series with limited time points that exhibit lead-lag effects, a scenario often encountered in biology, which inspired its development. 
These irregularities can be caused by missing values or study design restrictions. 
We believe that our method's adaptability renders it valuable for time series applications in various areas such as longitudinal gene expression, epigenetics, proteomics, cytometry, and microbiome studies.
It can be applied to a variety of problems including but not limited to ranking or clustering time series based on the degree of lead-lag effects, and interpolation or forecasting. 
The core feature of our methodology is the transform of the multi-output GP into a single-output GP, which facilitates the identification of parameters. This approach is grounded in theoretical support, affirming its validity as a covariance kernel for GP. This transform leverages the shared information among time series and enables simultaneous modeling of the time lag parameter $s$ and dissimilarity parameters $a$, thereby achieving desirable interpretability and expanding the application of the model. Additionally, our model can be extended to accommodate multiple time series, beyond just pairwise time series, enhancing its applicability and versatility.

The work here suggests some exciting future directions. First, when the sample size $n$ is small, the choice of the kernel is crucial in GP models. The extension to a non-stationary kernel allows for more flexibility when modelling non-stationary time series with more heterogeneity. The two main challenges are to find a valid non-stationary kernel with interpretable parameters and perform efficient and effective statistical inference. On the other hand, for time series data with large $n$, i.e., dense time points, the vanilla GP is computationally expensive with $\mathrm{O}(n^3)$ complexity~\citep{rasmussen2004gaussian}. In this situation, we can apply scalable GP methods such as nearest neighbor Gaussian processes (NNGP,~\citet{datta2016hierarchical}) or variational inference GP~\citep{tran2015variational} to reduce the complexity to $\mathrm{O}(n\log n)$ for large datasets. \changedAE{In addition, incorporating additional parameters, such as varying the variance $\sigma^2$ and lengthscale $b$ across different time series, could potentially enhance model performance.} \changedreviewerone{Finally, the question of whether the parameters in Theorem 3 can be consistently estimated remains open. Even for simpler kernels such as the Mat\'ern, consistent estimators for identifiable  parameters (e.g., $\sigma^2 b^{2\nu}$ and $\nu$) were only recently developed \citep{loh2021fixed,loh2023estimating}. For more complex kernels, like the ones proposed in this paper, constructing consistent estimators is an open and challenging problem.
}

\bibliographystyle{chicago}
\bibliography{biomsample.bib}

\appendix

\renewcommand\thefigure{S\arabic{figure}}
\renewcommand{\theHfigure}{Supplementary.\thefigure}
\renewcommand\theequation{S\arabic{equation}}

\setcounter{figure}{0}
\setcounter{equation}{0}

\section{Additional covariance functions for GPlag}\label{apdx:kernels}
All equations below can serve as covariance functions for GPlag, where $\sigma^2$ measures the spatial variance, $A$ measures the group dissimilarity, $S$ measures the time lag, $b>0$ measures the range/decay, $c>0$ measures the separability, and $\nu>0$ measures the smoothness of sample path. $K_\nu$ is the modified Bessel function of the second kind with degree $\nu$. 
\begin{align}
	K((t,l),(t',l'))&=
	\frac{2\sigma^2c^{1/2}\left\{\frac{b}{2}\left(\frac{a_{ll'}^2+1}{a_{ll'}^2+c}\right)^{1/2}|t-t'+s_l-s_{l'}|\right\}^\nu }{(a_{ll'}^2+1)^\nu(a_{ll'}^2+c)^{1/2}\Gamma(\nu)}K_\nu\left(b\left(\frac{a_{ll'}^2+1}{a_{ll'}^2+c}\right)^{1/2}|t-t'+s_l-s_{l'}|\right),\label{eqn:matern}\\
	K((t,l),(t',l'))&=
	\frac{\sigma^2c^{1/2}}{(a_{ll'}^2+1)^{1/2}(a_{ll'}^2+c)^{1/2}}\exp\left\{-b\left(\frac{a_{ll'}^2+1}{a_{ll'}^2+c}\right)^{1/2}|t-t'+s_l-s_{l'}|\right\},\label{eqn:matern1/2}\\
	K((t,l),(t',l'))&=\frac{\sigma^2}{(a_{ll'}^2+1)^{1/2}}\exp\left\{-\frac{b|t-t'+s_l-s_{l'}|}{(a_{ll'}^2+1)^{1/2}}\right\},\label{eqn:laplace}\\
	K((t,l),(t',l'))&=\frac{\sigma^2(a_{ll'}^2+1)^{1/2}}{(a_{ll'}^2+1)+b(t-t'+s_l-s_{l'})^2},\label{eqn:rational_quad}\\
	K((t,l),(t',l'))&=\sigma^2\exp\left\{-a_{ll'}^2-b(t-t'+s_l-s_{l'})^2-a_{ll'}^2(t-t'+s_l-s_{l'})^2\right\},\label{eqn:complex_exponential}
\end{align}

As an illustration, \Cref{eqn:matern} is the product of a kernel in Gneiting family with $\phi(t)=\frac{(\tilde{b}t^{1/2})^\nu}{2^{\nu-1}\Gamma(\nu)}K_\nu(\tilde{b}t^{1/2})$, $\psi(t)=(a^2t+c)/c(a^2t+1)$ in \Cref{cly:kernel_family} with $\tilde{b}=b/\sqrt{c}$, and a simple kernel $\frac{1}{(a_{ll'}^2+1)^{\nu+1/2}}$; \Cref{eqn:matern1/2} is a special case of \Cref{eqn:matern} when $\nu=1/2$. Moreover, if we set $c=1$, it becomes LExp; \Cref{eqn:laplace} can be derived by setting $\phi(t)=e^{-bt^{1/2}}$, $\psi(t)=a^2t+1$. \Cref{eqn:rational_quad} can be derived by setting $\phi(t)=(1+bt)^{-1}$, $\psi = a^2t+1$. \Cref{eqn:complex_exponential} is not in Gneiting family, but from \cite{cressie1999classes}, who proposed another spectral density based construction for nonseparable spatiotemporal kernels. 

In particular, when there are only two time series, the above kernels can be simplified by replacing $a_{ll'}$ by $\mathrm{\bf 1}_{\{l\neq l'\}}a^2$ and reparameterizing $s=s_2$.

\subsection{Proof of Theorem \texorpdfstring{\ref{thm:kernels} and \ref{thm:multiGPlag}}{1 and 2}}\label{apdx:mtsgp}

We prove the following lemma, which implies both \Cref{thm:kernels} and \Cref{thm:multiGPlag}. 
\begin{lemma}\label{thm:kernels_general}
	Let $K'$ be a semi-stationary kernel on $\RR\times \mathscr{C}$, where $\mathscr{C}=\{1,\cdots,L\}$, then the induced kernel defined as
	\begin{equation}\label{eqn:kernels_general}
		K((t,l),(t',l'))\coloneqq K'(t-t'+s_{l}-s_{l'},l,l')
	\end{equation}
	is positive definite, and thus serves as a covariance kernel for a GP, where $s_1$=0, $s_l$ measures the time lag with respect to the first time series.
\end{lemma}
\begin{proof}
	It suffices to show that given any set of observations $\{(t_i,c_i)\}_{i=1}^n$, the covariance matrix ${\Sigma}_{ij}={K}((t_i,c_i),(t_j,c_j))$ is positive definite. We start with a ``fake" set of observations, denoted by $\{(\widetilde{t}_i,c_i)\}_{i=1}^n$ induced by $\{(t_i,c_i)\}_{i=1}^n$ defined as $\widetilde{t}_{i}\coloneqq t_i+s_{c_i}.$ By the definition of ${K}$, we observe that
	\begin{align*}
		{\Sigma}_{ij}&={K}((t_i,c_i),(t_j,c_j))= K'(t_i-t_j+s_{c_i}-s_{c_j},c_i,c_j) = K'(\widetilde{t}_i-\widetilde{t_j},c_i,c_j)\eqqcolon \Sigma'_{ij},
	\end{align*}
	where $\Sigma'$ is the covariance matrix derived from $K'$, a positive definite kernel, evaluated on the ``fake" set of observations, hence is positive definite. As a result, $K$ is positive definite for any time lags $\{s_i\}_{l=1}^L$. 
\end{proof}
\Cref{thm:kernels} can be proved by setting $L=2$ and $s=s_2$.

\Cref{thm:multiGPlag} can be proved by setting $K'$ as
\begin{align*}
	K(t,l,l')&=\frac{\sigma^2}{
		(a_{ll'}^2+1)^{1/2}}e^{-\frac{bt^2}{a_{ll'}^2+1}},\\
	K(t,l,l')&=\frac{\sigma^2}{
		(a_{ll'}^2+1)}e^{-b\left|t\right|},\\
	K(t,l,l')&=\frac{\sigma^22\left\{\frac{b}{2}|t|\right\}^\nu }{(a_{ll'}^2+1)^{\nu+1/2}\Gamma(\nu)}K_\nu\left(b|t|\right).
\end{align*}

\subsection{Proof of \texorpdfstring{\Cref{thm:iden}}{3}}\label{apdx:iden}
A powerful tool to study equivalence of Gaussian processes is the spectral density:

\begin{definition}\label{def:spec}
	The spectral density $f$ associated with stationary Gaussian process $K$ in domain $\RR$ is defined as
	$$f(\omega) \coloneqq \frac{1}{\sqrt{2\pi}}\int_\RR e^{-i\omega t}K(t)\mathrm{d}t.$$
\end{definition}
\begin{remark}
	There are multiple commonly used definitions for spectral density. For example, $e^{-i\omega t}$ is sometimes replaced by $e^{-2\pi i \omega t}$ and/or without $\frac{1}{\sqrt{2\pi}}$. Throughout this article, we stick to the above definition and omit some absolute multiplicative constants for simplicity.
\end{remark}
The following lemma provides the spectral densities of the kernels introduced in Section \ref{sec:model}.
\begin{lemma}\label{lem:spec}
	The spectral densities of LRBF, \changedreviewerone{LExp and LMat} and  are given by
	\begin{align*}
		\rho_{ll'}(\omega) &= \frac{\sigma^2}{\sqrt{b}}e^{ i\omega(s_l-s_{l'})}\exp\left\{-\frac{(a_{ll'}^2+1)\omega^2}{b}\right\},\\
		\rho_{ll'}(\omega)&=\frac{\sigma^2b}{(a_{ll'}^2+1)(b^2+\omega^2)}e^{  i \omega(s_l-s_{l'})},\\
		\rho_{ll'}(\omega)&=\frac{\sigma^2b^{2\nu}}{(a_{ll'}^2+1)^{\nu+1/2}(b^2+\omega^2)^{\nu+1/2}}e^{  i \omega(s_l-s_{l'})}.
	\end{align*}
\end{lemma}

\begin{proof}
	First we calculate the spectral density of LRBF. When $l=l'$, the kernel becomes standard RBF kernel, so the spectral density is given by \citep{stein1999interpolation}:
	\begin{align*}
		\rho_{ll}(\omega) &= \frac{\sigma^2}{\sqrt{b}}e^{-\omega^2/b}.
	\end{align*}
	When $l\neq l'$, we replace $b$ in $\rho_{ll}$ by $\frac{b}{a_{ll'}^2+1}$ and due to  the time shifting property of Fourier transform, we have
	$$\rho_{ll'}(\omega) = \frac{\sigma^2}{\sqrt{b}}e^{ i\omega(s_l-s_{l'})}\exp\left\{-\frac{(a_{ll'}^2+1)\omega^2}{b}\right\},$$.
	
	Then we move to LMat. Similarly, first let $l=l'$, the kernel comes standard Mat\'ern kernel and the spectral density is \citep{stein1999interpolation}
	\begin{align*}
		\rho_{ll}(\omega) &=  \frac{\sigma^2b^{2\nu}}{(b^2+\omega^2)^{\nu+1/2}}.
	\end{align*}
	For arbitrary $l$ and $l'$, by the same trick, we have
	\begin{align*}
		\rho_{ll'}(\omega) 
	& = \frac{\sigma^2b^{2\nu}}{(a_{ll'}^2+1)(b^2+\omega^2)^{\nu+1/2}}e^{  i\omega(s_l-s_{l'})}.
\end{align*}
The spectral density of LExp is obtained by setting $\nu=\frac{1}{2}$ in the above equation.
\end{proof}

Then we introduce a test, known as the integral test, to determine whether two GPs are equivalent based on their spectral densities. We start with classical GPs on $\RR^p$.

\begin{lemma}[Integral test for univariate GP\citep{stein1999interpolation}]\label{lem:int_test_1}
Let $\rho_1$ and $\rho_2$ be spectral densities of two GPs $K_1$ and $K_2$ on domain $\RR^p$. Then $K_1\equiv K_2$ if
\begin{enumerate}
	\item There exists $\alpha>0$ such that $\rho_1(\omega)\|\omega\|^\alpha$ is bounded away from $0$ and $\infty$ as $\|\omega\|\to\infty$. 
	\item There exists $\delta>0$ such that $\int_{\|\omega\|>\delta}\left(\frac{\rho_2(\omega)}{\rho_1(\omega)}-1\right)^2\mathrm{d}\omega<\infty$.
\end{enumerate}
\end{lemma}
Lemma \ref{lem:int_test_1} is not directly applicable to our GPlags since our domain is $\RR\times \mathscr{C}$ where $\mathscr{C}$ is the index set of time series. However, a GPlag is equivalent to a multivariate GP with $L=|\mathscr{C}|$ outputs, defined as below:
\begin{definition}\label{def:multivariateGP}
A $L$-variate Gaussian process $f(t)\in\RR^L$ is characterized by its cross-covariance function $\widetilde{K}: \RR\times\RR\to \RR^{L}\times \RR^L: \cov(f(t),f(t'))=\widetilde{K}(t,t')$. 
\end{definition}

\begin{lemma}[Equivalence between GPlags and multivariate GPs \cite{li2021multi}]
A $L$-group GPlag on $\RR\times\mathscr{C}$ is equivalent to a $L$-variate GP on $\RR$, where $K$ and $\widetilde{K}$ are related by the following equation:
$$K((t,l),(t',l')=(\widetilde{K}(t,t'))_{ll'}.$$
\end{lemma}
The above Lemma connects multivariate time series and multiple time series. 

Let $K$ and $\widetilde{K}$ be two LMat kernels with parameters $(\sigma^2,b,A,S,\nu)$ and $(\widetilde{\sigma}^2,\widetilde{b},\widetilde{A},\widetilde{S},\nu)$. We first simply the notation as $\alpha = \frac{\sigma^2}{\pi(a^2_{ll'}+1)}$, $s=s_l-s_{l'}$, $\beta=b^{2\nu}$, then we have
$$\rho_{ll'}(\omega) = \frac{\alpha \beta}{(b^2+\omega^2)^{\nu+1/2}}e^{  i \omega s}.$$

Observe that $|\rho_{ll'}(\omega)|=\frac{\alpha \beta}{(b^2+\omega^2)^{\nu+1/2}}$ is independent of $l$ and $l'$, and $|\rho_{ll'}(\omega)||\omega|^{2\nu+1}$ is bounded away from $0$ and $\infty$ as $\omega\to\infty$.
Then we introduce a test to determined whether two GPlags are equivalence.
\begin{lemma}[Integral test for multivariate GPs tailored for LMat \citep{bachoc2022asymptotically}]
Let $\rho$ and $\widetilde{\rho}$ be spectral matrices of two LMats $K$ and $\widetilde{K}$, then $K\equiv \widetilde{K}$ if there exists $\delta>0$ such that
$$\int_{|\omega|>\delta}\left|\frac{\rho_{ll'}(\omega)}{\widetilde{\rho}_{ll'}(\omega)}-1\right|^2 \mathrm{d}\omega <\infty,~\forall l,l'=1,\cdots,L.$$

\end{lemma}

Now we can prove \Cref{thm:iden}. 
\begin{proof}[Proof of \Cref{thm:iden}]
Observe that
\begin{align*}
	\left|\frac{\rho_{ll'}(\omega)}{\widetilde{\rho}_{ll'}(\omega)}-1\right|^2 &= \left|\frac{\frac{\alpha b^{2\nu}}{b^2+\omega^2}e^{  i \omega s}}{\frac{\widetilde{\alpha} \widetilde{b}^{2\nu}}{\widetilde{b}^2+\omega^2}e^{  i \omega \widetilde{s}}}-1\right|^2\\
	& = \left|e^{  i \omega (s-\widetilde{s})}\frac{\alpha b^{2\nu} (\widetilde{b}^2+\omega^2)^{\nu+1/2}}{\widetilde{\alpha}\widetilde{b}^{2\nu}(b^2+\omega^2)^{\nu+1/2}}-1\right|^2\\
	& = \left|\frac{e^{  i \omega (s-\widetilde{s})}\alpha b^{2\nu} (\widetilde{b}^2+\omega^2)^{\nu+1/2}-\widetilde{\alpha}\widetilde{b}^{2\nu}(b^2+\omega^2)^{\nu+1/2}}{\widetilde{\alpha}\widetilde{b}^{2\nu}(b^2+\omega^2)^{\nu+1/2}}\right|^2\\
	& \leq \left|\frac{e^{  i \omega (s-\widetilde{s})}\alpha b^{2\nu} (\widetilde{b}^2+\omega^2)^{\nu+1/2}-\widetilde{\alpha}\widetilde{b}^{2\nu}(b^2+\omega^2)^{\nu+1/2}}{\widetilde{\alpha}\widetilde{b}^{2\nu}\omega^{2\nu+1}}\right|^2\\
	&  = \left|e^{  i \omega (s-\widetilde{s})}\frac{\alpha b^{2\nu}}{\widetilde{\alpha}\widetilde{b}^{2\nu}} \left(1+\left(\nu+\frac{1}{2}\right)\widetilde{b}^2/\omega^2\right)-\left(1+\left(\nu+\frac{1}{2}\right)b^2/\omega^2\right)\right|^2\\
	& =\left|e^{  i \omega (s-\widetilde{s})}\frac{\alpha b^{2\nu}}{\widetilde{\alpha}\widetilde{b}^{2\nu}}-1+O(\omega^{-2})\right|^2
\end{align*}
As a result,
$\int_{\delta}^\infty \left|\frac{\rho_{ll'}(\omega)}{\widetilde{\rho}_{ll'}(\omega)}-1\right|^2<\infty$ if and only if 
\begin{equation}\label{eqn:E_micro}
	e^{  i \omega (s-\widetilde{s})}\frac{\alpha b^{2\nu}}{\widetilde{\alpha}\widetilde{b}^{2\nu}}=1.
\end{equation}

\changedreviewerone{Since Equation \eqref{eqn:E_micro} holds for any $l,l',\omega$, we set certain parameters to special values to isolate conditions on certain parameters.} First we let $l=l'$ so that $s =s_l-s_{l'}=0$ and $a_{ll'}=0$. As a result, Equation \eqref{eqn:E_micro} becomes 
\begin{equation}\label{eqn:E_sigma2b}
	\sigma^2b^{2\nu} = \widetilde{\sigma}^2\widetilde{b}^{2\nu}.
\end{equation}
Then let $\omega = 0$, and plug Equation \eqref{eqn:E_sigma2b} into Equation \eqref{eqn:E_micro}, we have
\begin{equation}\label{eqn:E_a}
	\frac{1}{a_{ll'}^2+1}= \frac{1}{\widetilde{a}_{ll'}^2+1},~\text{i.e.}~ a_{ll'}=\widetilde{a}_{ll'}.
\end{equation}
Finally, plug Equation \eqref{eqn:E_sigma2b} and \eqref{eqn:E_a} into Equation \eqref{eqn:E_micro}, we have $e^{  i \omega (s-\widetilde{s})}=1$, that is, $s_l-s_{l'}=\widetilde{s}_l-\widetilde{s}_{l'}$. Recall that $s_1 = \widetilde{s}_1=0$ by the definition, so we let $l'=1$ to get
\begin{equation}\label{eqn:E_s}
	s_l=\widetilde{s}_l,~l=1,\cdots,L.
\end{equation}

Combination Equations \eqref{eqn:E_sigma2b}, \eqref{eqn:E_a}, and \eqref{eqn:E_s}, we conclude that
$$K\equiv \widetilde{K} \Longleftrightarrow (\sigma^2b,A,S)=(\widetilde{\sigma}^2\widetilde{b},\widetilde{A},\widetilde{S}).$$

To show finite nature of the integral necessitating the condition of \eqref{eqn:E_micro}, we assume $L=2$ and $\nu=1/2$, i.e., there are only two time series with LExp kernel. In this situation, it suffices to show that if $(\sigma^2b,a,s)\neq (\tilde{\sigma}^2\tilde{b},\tilde{a},\tilde{s})$, then $K(\cdot;\sigma^2,b,a,s)\not\equiv K(\cdot;\tilde{\sigma}^2,\tilde{b},\tilde{a},\tilde{s})$. 

First, we assume that $a=\tilde{a}$ and $s=\tilde{s}$, but $\sigma^2b\neq \tilde{\sigma}^2\tilde{b}$. Let $\{\psi_j\}$ be an orthonormal basis of the Hilbert space generated by $K(\cdot;\sigma_0^2,b,a,s)$. Let $\sigma_0^2 = \tilde{\sigma}^2\tilde{b}/b$, so we have $\sigma_0^2b=\tilde{\sigma}^2\tilde{b}$ and $K(\cdot;\sigma_0^2,b,a,s)\equiv K(\cdot;\tilde{\sigma}^2,\tilde{b},a,s) $ by the first half of the proof of \Cref{thm:iden}. So it suffices to show $K(\cdot;\sigma^2,b,a,s)\not\equiv K(\cdot;\sigma_0^2,b,a,s)$, which differs only by a multiplicative constant, $\eta=\sigma^2/\sigma_0^2=\sigma^2b/\tilde{\sigma}^2\tilde{b}\neq 1$. As a result, $\sum_{j,k=1}(E_{K(\cdot;\sigma^2,b,a,s)}\psi_j\psi_k-E_{K(\cdot;\sigma_0^2,b,a,s)}\psi_j\psi_k)^2=\sum_{j=1}^\infty (\eta-1)^2=\infty$. By Ibragimov and Rozanov (2012, p.72), we conclude that $K(\cdot;\sigma^2,b,a,s)\not\equiv K(\cdot;\sigma_0^2,b,a,s)$.

Similarly, when $a\neq \tilde{a}$, we can prove the same result with $\eta = (\tilde{a}^2+1)/(a^2+1)$; while when $s\neq \tilde{s}$, the corresponding $\eta = e^{-b\tilde{s}}/e^{-bs}$. 
\end{proof}

\subsection{Detailed algorithm}\label{apdx:algorithm}

\Cref{alg:GPlag} and \Cref{alg:GPlag-multi} illustrates how to calculate the MLE for pairwise and three time series separately. 

\begin{algorithm}[!ht]
\caption{MLE process of GPlag optimization on pairwise time series}
\label{alg:GPlag}
\begin{algorithmic}
	\STATE {\bfseries Input:} $\{(t_i,c_i,y_i)\}_{i=1}^n, c_i\in \{1, 2\}, n = n_1+n_2$
	\STATE {\bfseries Output:} Parameter estimation: $a$, $b$, $s$, $\sigma^2$, $\tau^2$
	\STATE{\bfseries Step 1:} Initialization
	\STATE{\bfseries Step 2:} $\Sigma_{ij}=K((t_i,c_i),(t_j,c_j)) = \frac{\sigma^2}{
		(a^2\mathbf{1}_{\{c_i\neq c_j\}}+1)^{1/2}}e^{-\frac{b||t-t'-s||^2}{a^2\mathbf{1}_{\{c_i\neq c_j\}}+1}} + \tau^2\mathrm{I}_n$ 
	\STATE $~~~~~~~~~~~\log p(y | b, a, s, \sigma^2, \tau^2) = - \frac{1}{2}\log|\Sigma| - \frac{n}{2} \log(Y^\top \Sigma^{-1}Y), \Sigma_{ij} = K((t_i,c_i),(t_j,c_j))$ 
	\STATE{\bfseries Step 3:} Maximize the above log-likelihood by user selected optimizer such as Adam or L-BFGS-B.
	
\end{algorithmic}
\end{algorithm}

\begin{algorithm}[!ht]
\caption{MLE process of GPlag optimization on three time series}
\label{alg:GPlag-multi}
\begin{algorithmic}
	\STATE {\bfseries Input:} $\{(t_i,c_i,y_i)\}_{i=1}^n, c_i\in \{1,\dots, 3\}, n = n_1+n_2+n_3$
	\STATE {\bfseries Output:} Parameter estimation: $\{a_{ll'}\}_{l,l'=1,2,3}, b, s_1, s_2, s_3, \sigma^2, \tau^2$
	\STATE{\bfseries Step 1:} Initialization
	\STATE{\bfseries Step 2:} $\Sigma_{ij}=K((t_i,c_i),(t_j,c_j)) = \frac{\sigma^2}{
		(a_{c_ic_j}^2+1)^{1/2}}e^{-\frac{b||t_i-t_j-(s_{c_i}-s_{c_j})||^2}{a_{c_ic_j}^2+1}} + \tau^2\mathrm{I}_n$ 
	\STATE $~~~~~~~~~~~\log p(y | b, \mathbf{a}, \mathbf{s}, \sigma^2, \tau^2) = - \frac{1}{2}\log|\Sigma| - \frac{n}{2} \log(Y^\top \Sigma^{-1} Y), \Sigma_{ij} = K((t_i,c_i),(t_j,c_j))$ 
	\STATE{\bfseries Step 3:} Maximize the above log-likelihood under the constraints $a_{ll'}=0\Longleftrightarrow l=l',~a_{ll'}=a_{l'l},~a_{ll'}+a_{l'k}\geq a_{lk}$ for any $l,l',k=1,2,3$ and $s_1=0$, by user selected optimizer such as Adam or L-BFGS-B.
\end{algorithmic}
\end{algorithm}

\Cref{alg:GPlag-bayesian} illustrates the realization of the Bayesian inference procedure.
\begin{algorithm}[!ht]
\caption{Bayesian inference procedure of GPlag optimization on pairwise time series}
\label{alg:GPlag-bayesian}
\begin{algorithmic}
	\STATE {\bfseries Input:} $\{(t_i,c_i,y_i)\}_{i=1}^n, c_i\in \{1, 2\}, n = n_1+n_2$
	\STATE {\bfseries Output:} Parameter estimation: $a$, $b$, $s$, $\sigma^2$, $\tau^2$
	\STATE{\bfseries Step 1:} Prior distribution: $$\pi(a,b, s,\sigma^2,\tau^2) = IG(a|\alpha_a, \alpha'_a)IG(b|\alpha_b, \alpha'_b)N(s|\mu_s, V_s)IG(\sigma^2|\alpha_{\sigma}, \alpha'_{\sigma})IG(\tau^2|\alpha_{\tau}, \alpha'_{\tau}),$$
	\STATE{\bfseries Step 2:} Posterior distribution: $$\pi(a,b,s,\sigma^2,\tau^2|\{(t_i,c_i,y_i)\}_{i=1}^n)\propto \pi(a,b, s,\sigma^2,\tau^2)p(Y|\{(t_i,c_i,y_i)\}_{i=1}^n, a,b,s,\sigma^2,\tau^2),$$
	\STATE $~~~~~~~~~~~\log p(y | b, a, s, \mu) = - \frac{1}{2}\log|\Sigma| - \frac{n}{2} \log(Y^\top \Sigma^{-1}Y), \Sigma_{ij} = K((t_i,c_i),(t_j,c_j))$ 
	\STATE{\bfseries Step 3:} Sampling: Run No-U-Turn Sampler(NUTS) algorithm to generate samples from the above posterior distribution to achieve estimator uncertainty.
	
\end{algorithmic}
\end{algorithm}

where $IG$ denotes the inverse gamma distribution and $N$ denotes the normal distribution, $\alpha_a$, $\alpha'_a$, $\alpha_b$, $\alpha_b',\mu_s,V_s,\alpha_\tau,\alpha'_\tau$ are hyper-parameters in the prior. 

\section{Additional experimental details} 
\label{apdx:experiments}
\subsection{Simulation}\label{apdx:simulation}

\textbf{Parameter Estimation and Inference}
The variable $T$ with various sample size was generated within the range (-50, 50), and a random variation was added by sampling from a uniform distribution with a range of $\mathrm{unif}(-\frac{1}{4},\frac{1}{4})$. Three kernels were all served as a covariance function to generate pairwise time series from Gaussian process(GP) with $b$, $a$, $s$ set to $0.3, 1, 2$. When fitting GPlag, the parameters $b$, $a$, $s$ were initialized with 1, 1, 1, respectively. For both GPlag and Lead-Lag methods, the upper bound for the time lag was set to 4. The maximum likelihood estimation (MLE) was obtained using the L-BFGS-B algorithm. \changedreviewertwo{We plot the MLE of $b$ and $\sigma^2$, as well as $\sigma^2b$ (\Cref{fig:bsigma2}). The value of 
$\sigma^2b$ closely matches its true value, aligning with the theoretical predictions discussed in Theorem  \Cref{thm:iden}. Additionally, as $n$ increases, the variance decreases, demonstrating the consistency of the MLE.}


\changedreviewertwo{Furthermore, to evaluate the method's performance across different lengthscales $b$, and to examine their influence on the estimates of $a$ and $s$, we simulate $b=0.3, 5, 10$ for $n=50,100,200,300$ using the LExp kernel. The results, summarized in the \Cref{fig:variousb}, indicate that the value of $b$ does influence model optimization. Nonetheless, MLE of both $a$ and $a$ converge to their true values as the sample size increases.}

\begin{figure}[htbp]
\centering
\vspace{-10pt} 
\includegraphics[width=0.8\textwidth]{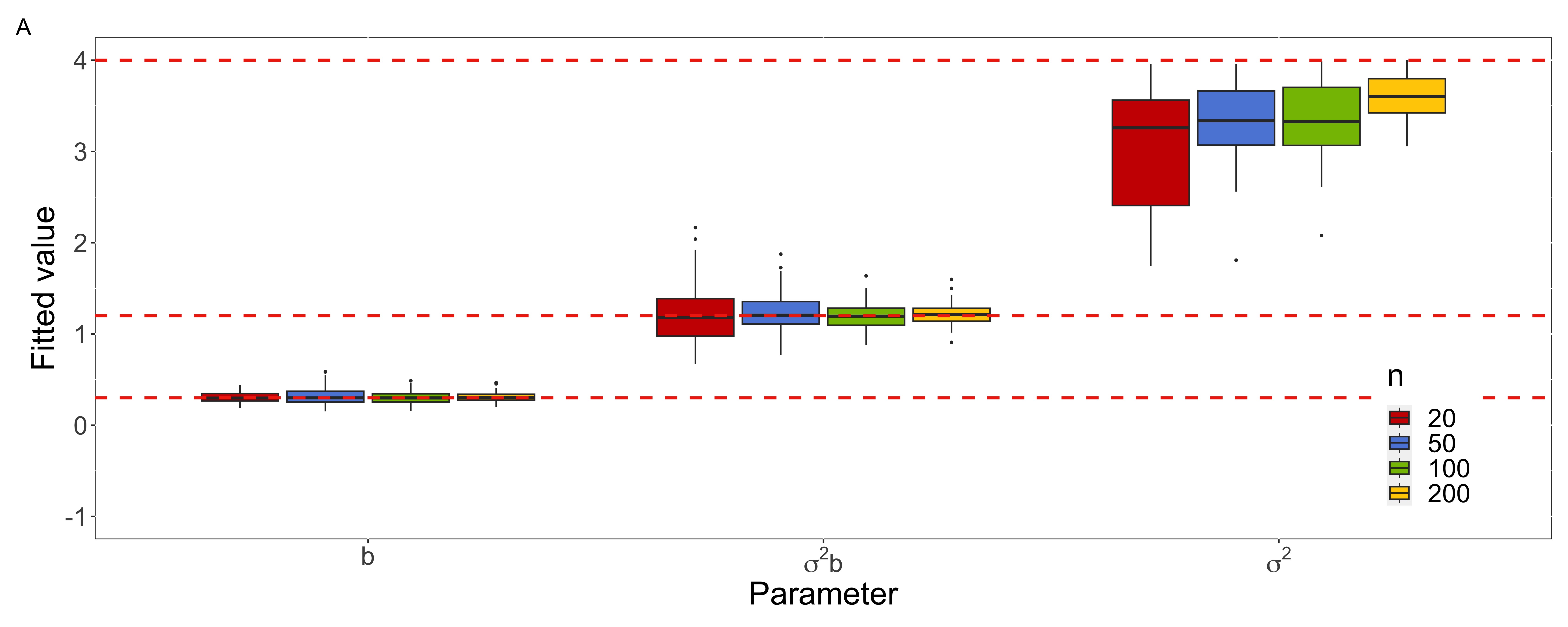}
\caption{ MLEs of GPlag parameters $b$ and $\sigma^2$ and $\sigma^2b$ in the LExp kernel with
	different n. The true parameter values were represented by the dashed red horizontal lines. All
	boxplots were from 100 replicates}
\vspace{-10pt} 
\label{fig:bsigma2}
\end{figure}

\begin{figure}[htbp]
\centering
\vspace{-10pt} 
\includegraphics[width=0.8\textwidth]{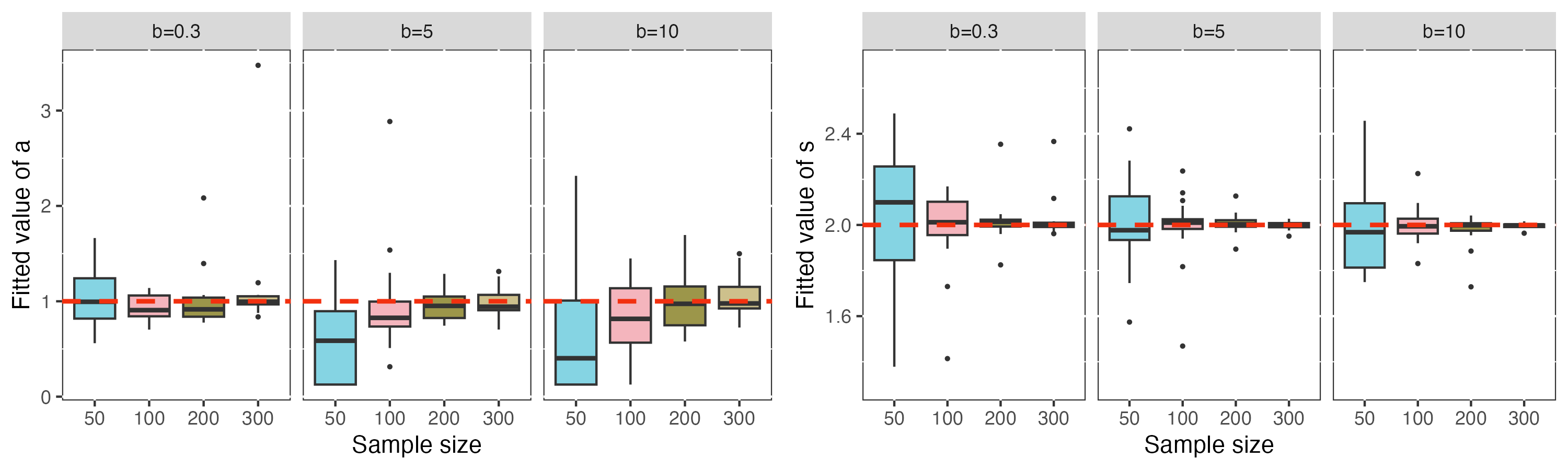}
\caption{ MLE of $a$ and $s$ in GPlag with data generated from LExp kernel with different lengthscales $b$ and sample size. The true parameter values were represented by the dashed red horizontal lines. All boxplots were from $20$ replicates}
\vspace{-10pt} 
\label{fig:variousb}
\end{figure}

\textbf{Prediction}
For the process that time series generated from the kernel LExp in \Cref{eqn:LExp}, the variable $T$ with $n_l = 50$ was generated within the range (-25, 25), and a random variation was added by sampling from a uniform distribution with a range of $\mathrm{unif}(-\frac{1}{4},\frac{1}{4})$. For the linear process with non-Gaussian noise, the variable $T$ with $n_l = 100$ was generated within the range (0, 100) with $s = 20$. 50\% data were randomly sampled as training sets in one time series, and the corresponding positions in the other time series were also treated as training sets. The remaining 50\% data were treated as testing sets. The initial, lower, and upper bound for the time lag was set to 1, -1, 5, respectively. The maximum likelihood estimation (MLE) was obtained using the L-BFGS-B algorithm. 

\begin{figure}[h]
\centering
\vspace{-10pt} 
\includegraphics[width=0.3\textwidth]{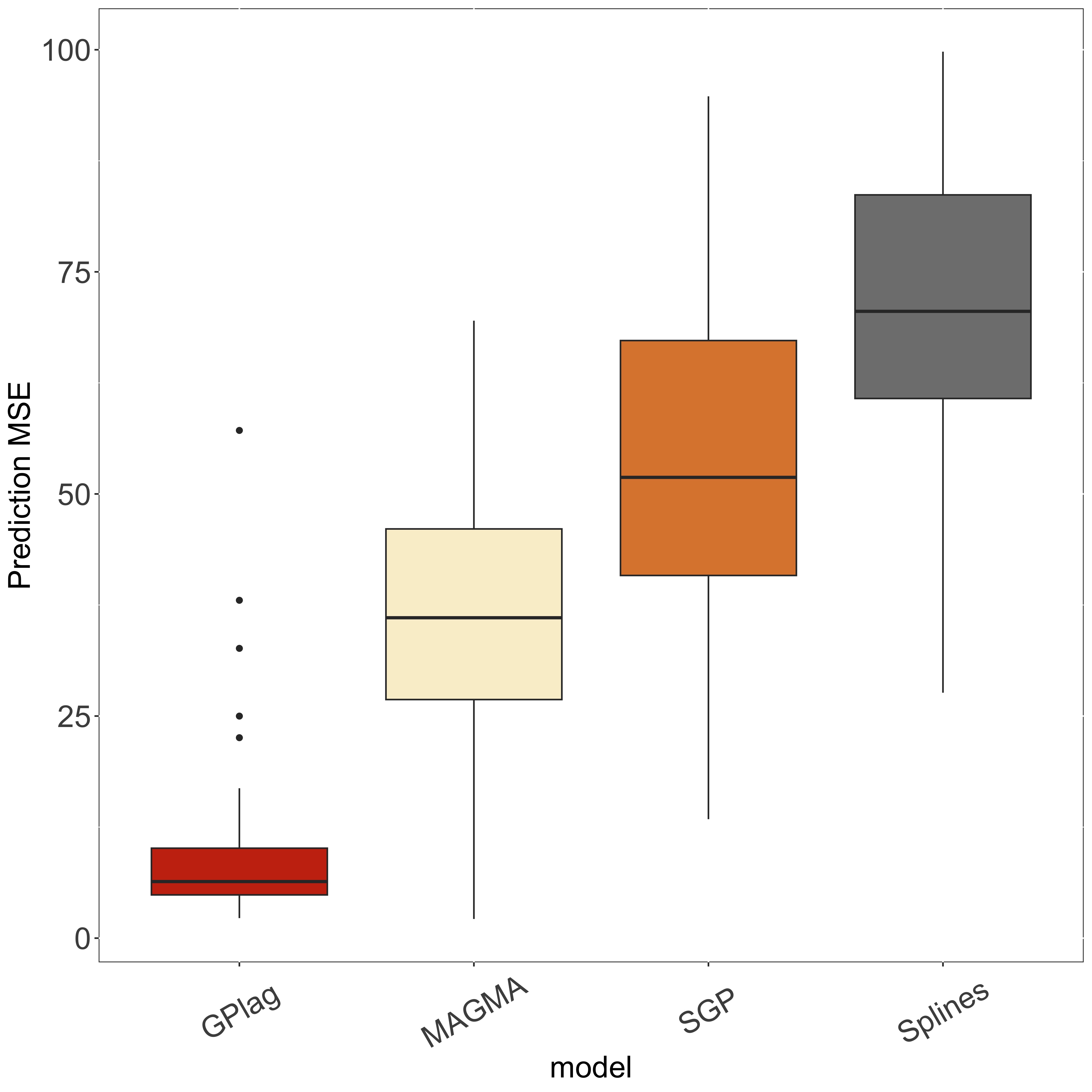}
\caption{Prediction accuracy comparing with baselines for pair time series generated from a linear process $y =  2t + 3 + 5\epsilon_1$. $\epsilon_1$ was sampled from t-distributed noise. Each time series has $n=100$ observations, whereas the time shift =20 between time series. Replicates 100 times.}
\vspace{-10pt} 
\label{fig:R2Q2}
\end{figure}

\textbf{Clustering or Ranking Time Series based on dissimilarity to target time series}
The variable $T$ with $n_l = 50$ for 10 time series (target time series + other 9 time series) was generated  regularly from (-2, 2). When fitting GPlag, the parameters $b$, $a$, $s$ were initialized with 1, 1, 0, respectively. The lower and upper bound for the time lag was set to -1 and 4, respectively. The maximum likelihood estimation (MLE) was obtained using the L-BFGS-B algorithm. \changedreviewertwo{We also replicated this experiment using LRBF and LMat kernels. For both kernels, the ARI and NMI scores were 1 (same as LExp kernel), and the inferred values of $a$ were again clearly separated into three distinct clusters \Cref{fig:cluster_rbf_mat}.}

\begin{review2_env}
\begin{figure}[h]
	\centering
	\vspace{-10pt} 
	\includegraphics[width=\textwidth]{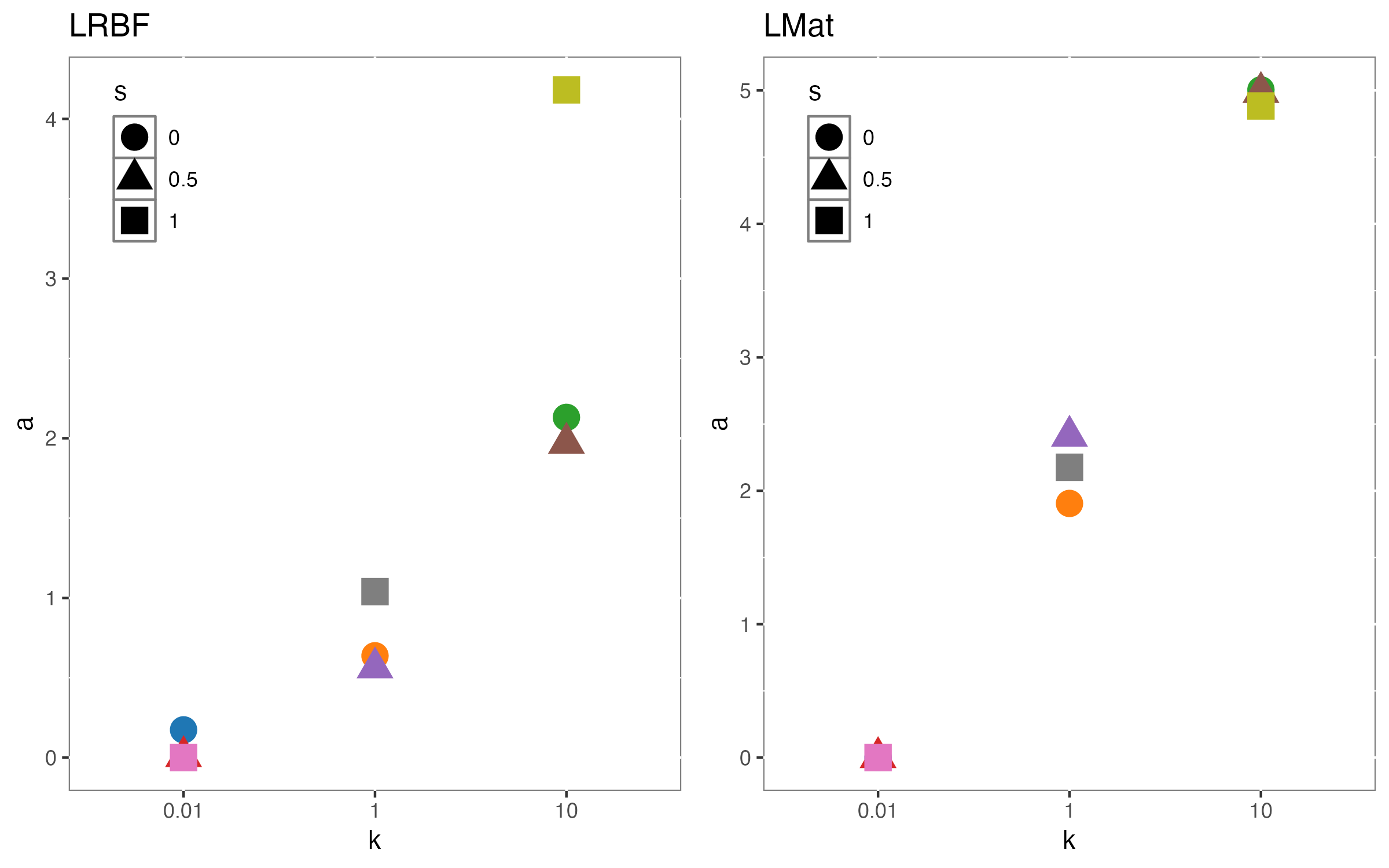}
	\caption{Inferred $a$ with LRBF and LMat kernels in the clustering experiment.}
	\vspace{-10pt} 
	\label{fig:cluster_rbf_mat}
\end{figure}
\end{review2_env}

\begin{review3_env}

\textbf{Visual example of kernels}

To better help the reader understand and choose the appropriate kernel, here we simulate data from LExp, LMat, and LRBF with various parameter. The simulated sample paths are shown in the following figures. The figures clearly demonstrate that a smaller $b$ indicates higher similarity between nearby points, while a smaller $a$ reflects greater similarity between the two time series.

\begin{figure}[h]
	\centering
	\vspace{-10pt} 
	\includegraphics[width=\textwidth]{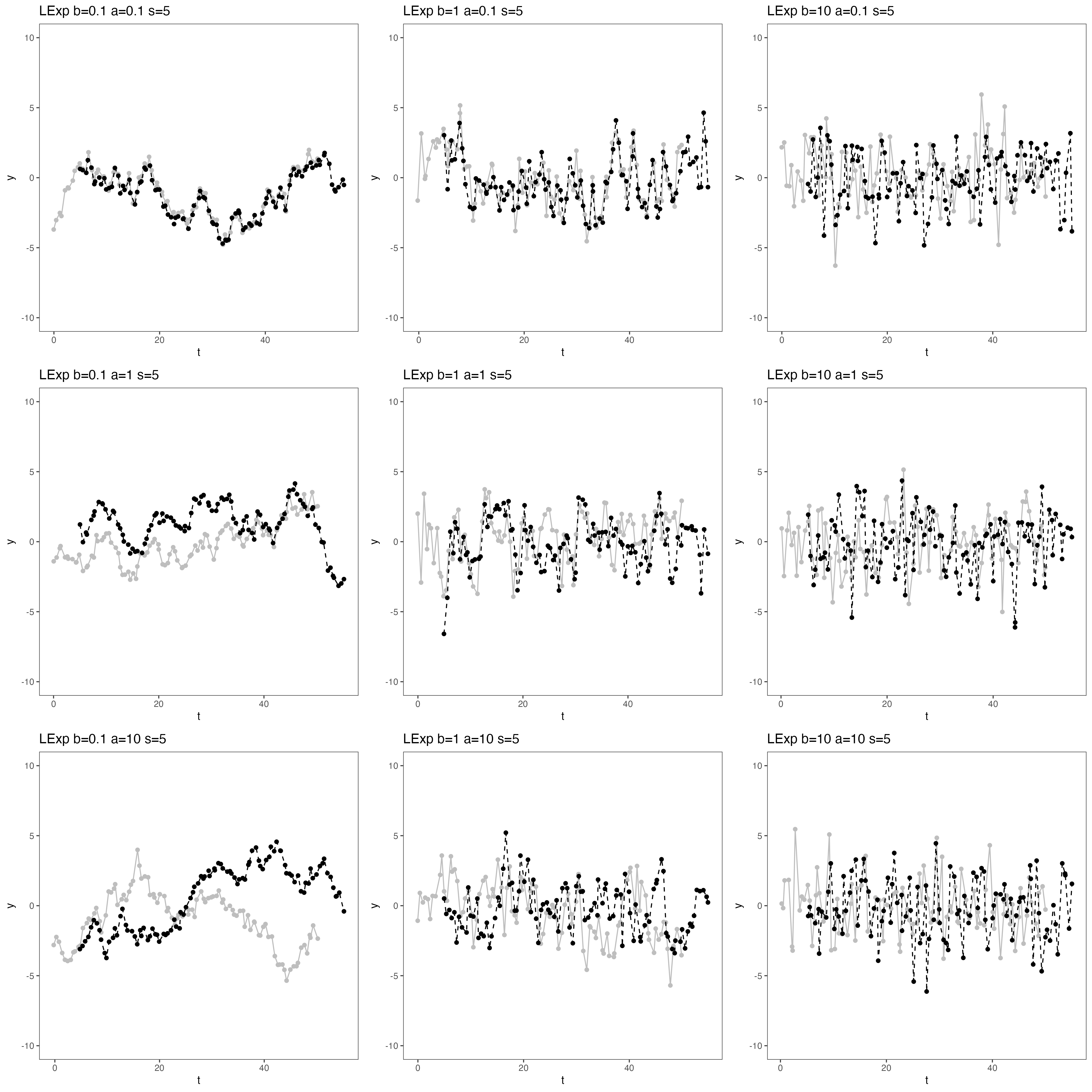}
	\caption{Simulated sample path for LExp kernel.}
	\vspace{-10pt} 
	\label{fig:sampp_exp}
\end{figure}

\begin{figure}[h]
	\centering
	\vspace{-10pt} 
	\includegraphics[width=\textwidth]{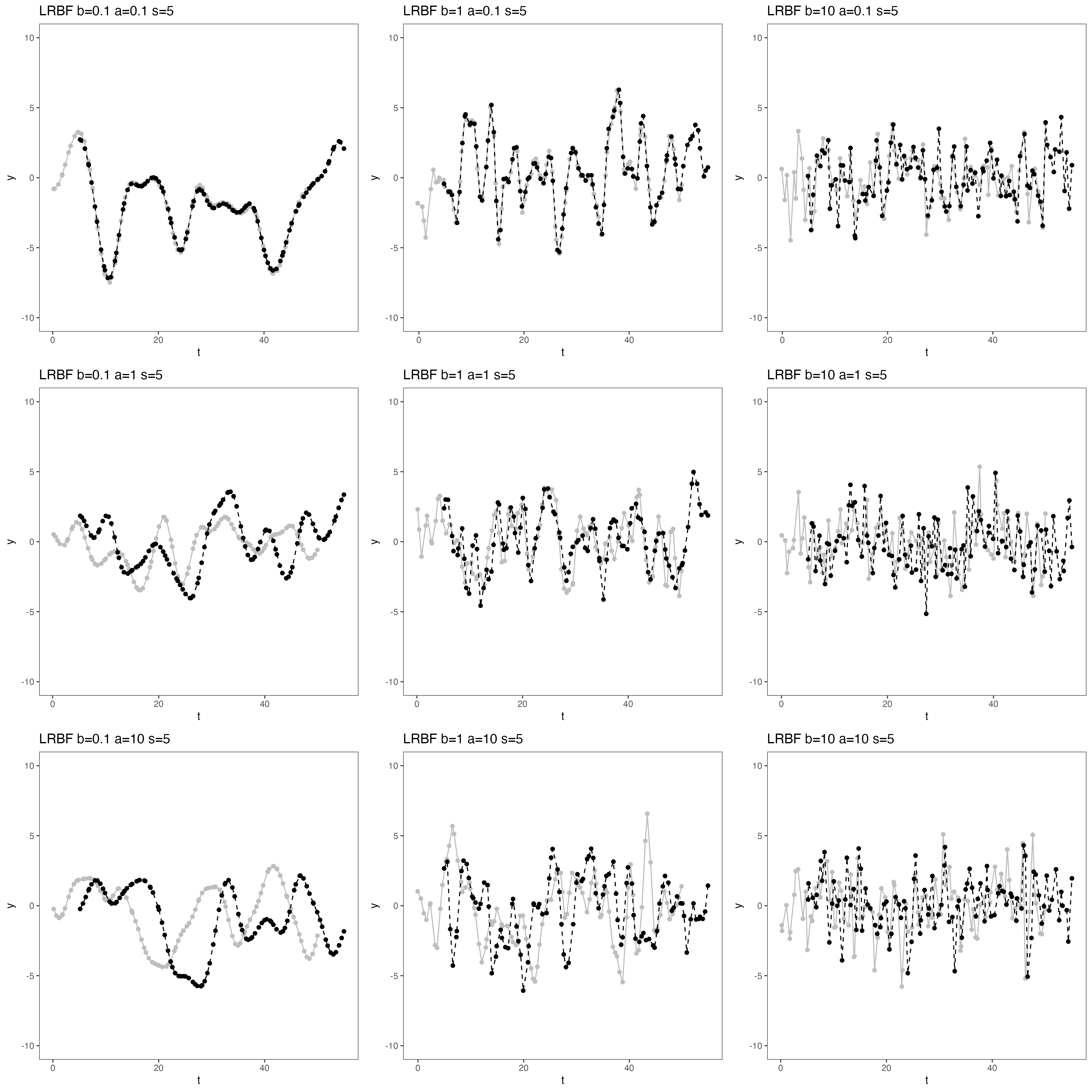}
	\caption{Simulated sample path for LRBF kernel.}
	\vspace{-10pt} 
	\label{fig:sampp_rbf}
\end{figure}

\begin{figure}[h]
	\centering
	\vspace{-10pt} 
	\includegraphics[width=\textwidth]{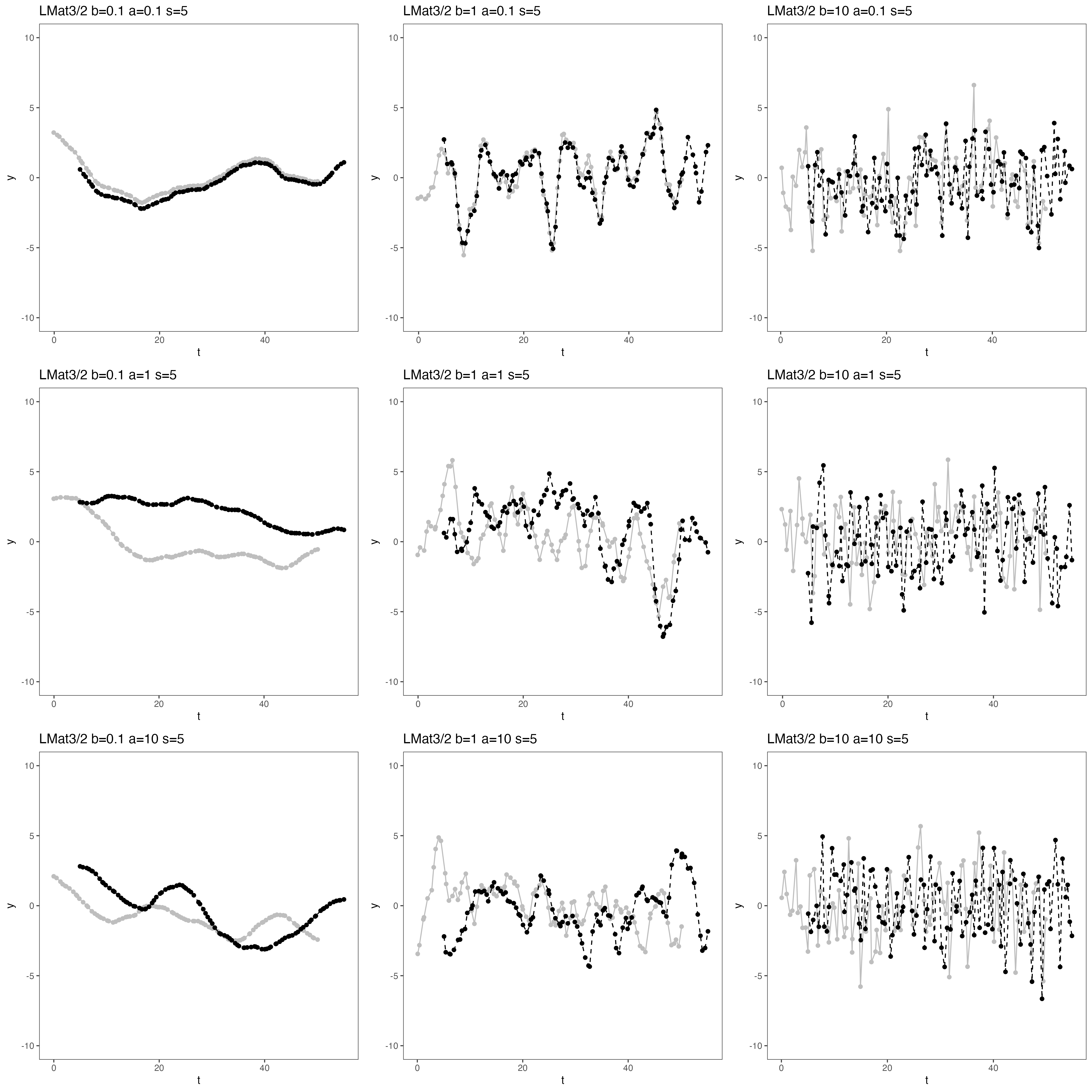}
	\caption{Simulated sample path for LMat ($\nu=\frac{3}{2}$) kernel.}
	\vspace{-10pt} 
	\label{fig:sampp_matern}
\end{figure}

\end{review3_env}

\subsection{Real-world applications}
In order to understand how enhancers regulate gene transcription, 
\citet{reed2022temporal} evaluated human macrophage activation after stimulating macrophage with LPS and interferon-gamma (IFN$_\gamma$). The data were collected at seven irregularly-spaced time points (0, 0.5, 1, 1.5, 2, 4, 6 h). At each time point, the 3D chromatin structure was profiled using \textit{in situ} Hi-C, putative enhancer activity using chromatin immunoprecipitation sequencing (ChIP-seq) targeting histone 3 lysine 27 acetylation (H3K27ac), chromatin accessibility using assay for transposase-accessible chromatin using sequencing (ATAC-seq), and gene expression using RNA sequencing (RNA-seq). Hi-C scale was indicated in KR-normalized counts, H3K27ac and RNA-seq using variance stabilizing transformation (VST) from the \emph{DESeq2} package \citep{love2014moderated}. 

We characterize synchronization between activity at enhancers and transcription at promoters allowing for time lags. 
Time lags help to better model potential causal relationships between chromatin accessibility and gene transcription.
Usually, changes in acetylation at distal enhancers precede changes in gene
expression \citep{reed2022temporal}, so we assume the time lag parameter $s$ is positive and set its upper bound to 2h. In total, we have 22,813 enhancer and gene pairs. Since we believe stimulation would activate enhancer and gene changes, we removed the gene and enhancers whose range of normalized expression is smaller than 1. 4,776 non-static pairs remained for analysis.
For Hi-C value, expected counts are determined using distance decay from the diagonal and the total read depth of the bins, so the resulting value is independent of distance. 

The GPlag model employed the LRBF kernel in \Cref{eqn:LRBF}, with time initialization obtained from TLCC and parameter $b$ initialization from GP modeling on gene expression. The optimization process utilized the L-BFGS-B algorithm. 
\end{document}